\documentclass[10pt,twocolumn,letter]{IEEEtran}

\usepackage{amsmath,graphicx}
\usepackage{subfigure}
\usepackage{amssymb}
\usepackage{epsfig}
\usepackage{algorithm}
\usepackage{caption}
\usepackage{ragged2e}
\usepackage{epstopdf}
\usepackage{color}
\usepackage{balance}
\usepackage{cite}

\newcommand{\beq}{\begin{equation}}
\newcommand{\eeq}{\end{equation}}

\newcommand{\ds}{\displaystyle}

\def\bI{\mbox{\boldmath $I$}}

\def\b0{\mbox{\boldmath $0$}}

\def\bp{\mbox{\boldmath $p$}}

\def\bt{\mbox{\boldmath $t$}}

\def\by{\mbox{\boldmath $y$}}

\newtheorem{theorem}{Theorem}

\newenvironment{proof}[1][Proof]{\noindent \textbf{#1.} }{\qedsymbol}
\newcommand{\qedsymbol}{\hspace{\fill}\rule{1.5ex}{1.5ex}}

\title{Joint Optimization of Radio Resources and Code Partitioning in Mobile Edge Computing}

\author{Paolo~Di Lorenzo,~\IEEEmembership{Member,~IEEE}, Sergio~Barbarossa,~\IEEEmembership{Fellow,~IEEE}, and Stefania Sardellitti,~\IEEEmembership{Member,~IEEE},\vspace{-.3cm} 
\thanks{Di Lorenzo is with the Dept. of Engineering, University of Perugia, Via G. Duranti 93, 06125, Perugia, Italy; Email: \texttt{paolo.dilorenzo@unipg.it}. Barbarossa and Sardellitti are with the Department of Information Engineering, Electronics, and Telecommunications, Sapienza university of Rome, Via Eudossiana 18, 00184, Rome, Italy; Email: Email: \texttt{sergio.barbarossa@uniroma1.it}, \texttt{stefania.sardellitti@uniroma1.it}. This work has been supported by TROPIC Project, Nr. ICT-318784.}}

\begin{document}

\maketitle

\begin{abstract} 
The aim of this paper is to propose a computation offloading strategy for mobile edge computing. We exploit the concept of \textit{call graph}, which models a generic computer program as a set of procedures related to each other through a weighted directed graph. Our goal is to derive the optimal partition of the call graph establishing which procedures are to be executed locally or remotely. The main novelty of our work is that the optimal partition is obtained jointly with the selection of radio parameters, e.g., transmit power and constellation size, in order to minimize the energy consumption at the mobile handset, under a latency constraint taking into account transmit time and execution time. We consider both single and multi-channel transmission strategies and we prove that a globally optimal solution can be achieved in both cases. Finally, we propose a sub-optimal strategy aimed at solving a relaxed version of the original problem in order to tradeoff complexity and performance of the proposed framework. Finally, several numerical results illustrate under what conditions in terms of call graph topology, communication strategy, and computation parameters, the proposed offloading strategy provides large performance gains.
\end{abstract}

\vspace{-.5cm}
\section{Introduction}

Computation offloading has attracted a lot of research efforts as a way to augment the capabilities of resource-constrained and energy-hungry mobile handsets by migrating computation to more resourceful servers. Offloading is useful either to enable smartphones to run more and more sophisticated applications, while meeting strict delay constraints, or to prolong the battery lifetime by limiting the energy spent at the mobile handset to run a given application under application-dependent delay constraints. In the current mobile handset market, user demand is increasing for several categories of energy-hungry applications. Video games are becoming more and more popular and they have large energy requests. Mobile users are also increasingly watching and uploading streaming video from YouTube \cite{Lewin1}. Furthermore, mobile devices are increasingly equipped with new sensors that produce continuous streams of data about the environment. New applications relying on continuous processing of the collected data are emerging, such as car navigation systems, pedometers, and location-based social networking, just to name a few. All these applications are severely energy demanding for today's mobile handsets. On the other hand, the trend in battery technology makes it unlikely that the energy problem will disappear in the near future. Indeed, recent researches on the adoption of new battery technologies, e.g., new chemicals, are unlikely to significantly improve battery lifetime \cite{Powers,Palacin}. A possible strategy to overcome the above energy/computation bottleneck consists in enabling resource-constrained mobile devices to offload
their most energy-consuming tasks to nearby more resourceful servers. This strategy has a long history and is reported
in the literature under different names, such as cyber foraging \cite{Sharifi}, or computation offloading \cite{Kumar-Liu-Lu-Bhargava}.
A great impulse to computation offloading has come through the introduction of cloud computing \cite{Hayes}. One of the key features of cloud computing is {\it virtualization}, i.e. the capability to run multiple operating systems and multiple applications over the same machine (or set of machines), while guaranteeing isolation and protection of the programs and their data. In particular, mobile cloud computing (MCC) \cite{Dinh-Lee-Niyato-Wang}, \cite{Fernando-Loke-Rahayu}, \cite{SB-SS-PDL2}, makes possible for mobile users to access cloud resources, such as infrastructures, platforms, and software, on-demand through a wireless access.
Several works addressed mobile computation offloading, such as \cite{Li-Wang-Xu,Wang-Li,Yang-Ou-Chen,Chun-Maniatis,Chun-Ihm-Maniatis-Naik,Tang,Kwon,Gao-He-Liu-Li-Jarvis,Wen-Zhang-Luo,Barb-Sard-Dilo,Barb-Dilo-Sard}. The offloading decisions are usually made by analyzing parameters including bandwidths \cite{Wolsky-Gurun-Krintz-Nurmi}, server speeds \cite{Kao}, available memory, server loads \cite{Barb-Sard-Dilo}, and the amounts of data exchanged between servers and mobile systems. The solutions include {\it code partitioning} programs \cite{Li-Wang-Xu,Wang-Li,Wen-Zhang-Luo, Maui,Chen-Kang-Kandemir-Vijaykrishnan-Irwin-Chandramouli, Hong-Kumar-Lu, Nimmagadda-Kumar-Lee, Ou-Yang-Liotta-Lu, Xian-Lu-Li,Liu,Yang13}, and predicting parametric variations in application behavior and execution environment \cite{Wolsky-Gurun-Krintz-Nurmi,Huerta-Canepa-Lee}. Offloading typically requires code partitioning \cite{Liu}, aimed to decide which parts of the code should run locally and which parts should be offloaded, depending on contextual parameters, such as computational intensity of each module, size of the program state to be exchanged to transfer the execution from one site to the other, battery level, delay constraints, channel state, and so on. Program partitioning may be {\it static} or {\it dynamic}. Static partitioning has the advantage of requiring a low overhead during execution, but it works well only if the parameters related to the offloading decisions are accurately known in advance or predicted, see e.g. \cite{Li-Wang-Xu,Kumar-Lu, Balan,Chen-Kang-Kandemir}. In contrast, enforcing a dynamic decision mechanism makes possible to adapt the method to different operating conditions and to cope with different degrees of uncertainties \cite{Wang-Li,Kwon,Xian-Lu-Li,Huerta-Canepa-Lee,Kovachev-Yu-Klamma,Huang-Wang-Niyato}. A dynamic approach to computation offloading consists in establishing a rule to decide which part of the computations may be advantageously offloaded, depending on channel conditions, server state, delay constraints, and so on. Specific examples of dynamic mobile computation offloading techniques are: MAUI \cite{Maui}, ThinkAir \cite{Kosta_2012}, and Phone2Cloud \cite{Xia2014}. Finally, mobile cloud computing strategies in a multiuser scenarios were considered in \cite{SB-SS-PDL2},\cite{SB-SS-PDL,Labidi,Sard-Scut-Barb,Chen,Cardellini}.


One of the limitations of today mobile cloud computing is the latency experienced in the propagation of information through a wide area network (WAN). As pointed out in \cite{Maui}, \cite{Satyanarayanan-Bahl-Caceres-Davies}, \cite{Barbera-Kosta-Mei-Stefa}, humans are acutely sensitive to delay and jitter, and it is very difficult to control these parameters at WAN scale: As latency increases, interactive response suffers. Furthermore, mobile access through macro base stations (MBS) might require large power consumptions.
To overcome this limitation, the authors of \cite{Satyanarayanan-Bahl-Caceres-Davies} introduced the concept of {\it cloudlet}, i.e. the possibility for the mobile handsets to access nearby static resourceful computers, linked to a distant cloud through high speed wired connections.  In \cite{Satyanarayanan-Bahl-Caceres-Davies}, cloudlets are imagined to be deployed much like Wi-Fi access points today. Comparisons of energy expenditure for computation offloading using 3G and Wi-fi systems are given in \cite{Maui} and \cite{Barbera-Kosta-Mei-Stefa}, where it is shown how the benefits of offloading are most impressive when the remote server is nearby, such as on the same LAN as the Wi-Fi access point, rather than using 3G technology. Within the European project TROPIC \cite{TROPIC} we proposed to endow small cell LTE base stations, also known as Small Cell e-Node B (SCeNB), with, albeit limited, cloud functionalities. In this way, one can bring computing resources very close to the mobile user equipments (MUE). Considering the dense deployment of radio access points (RAP) foreseen in 5G, if these RAP are endowed with computing capabilities, there is the possibility of a very effective provision of computing capabilities in close proximity to the mobile user, with advantages over Wi-Fi access in terms of Quality-of-Service, through a single technology system (no need for the MUE's to switch between cellular and Wi-Fi standards). Merging information technologies, wireless access and cloud computing within a holistic framework is indeed one of the key steps for an effective 5G deployment, as evidenced by the new standardization group on the so called Mobile-Edge Computing (MEC), recently launched by the European Telecommunications Standards Institute (ETSI). The aim of MEC is precisely to bring computing capabilitites at the edge of the network,  in close proximity to mobile subscribers, thus offering a service environment characterized by very low latency and high rate access \cite{ETSI}.
%

\textbf{Contribution:} Our goal is to propose a computation offloading method operating in the aforementioned MEC framework. The key basic assumption, justified within the MEC framework, is that the allocation of radio resources takes into account features pertaining to the application layer. The proposed approach optimizes code partitioning and radio resource allocation {\it jointly}, in order to minimize the overall energy consumption at the mobile handset, under a latency constraint dictated by the application requirements. To find the optimal code partitioning, we exploit the concept of \textit{call graph} of a program \cite{Ryder,Grove-Chambers}, which models the relations between procedures  of a computer program through a directed acyclic graph. Our objective is to find the code partition and the radio resource allocation (power and constellation size) that minimize the energy consumption at the mobile side, under a latency constraint that incorporates communication and computation times. We consider both a single channel and a multi-channel transmission strategies. We provide theoretical results proving the existence of a unique solution of the resource allocation problem. The theoretical findings are then corroborated with numerical results, illustrating  under what kind of call graph's topologies, radio channel conditions, and communication strategies, the proposed method can provide a significant energy saving at the mobile handset. In principle, the joint optimization problem has combinatorial complexity. However, we first illustrate that several graph configurations can be apriori discarded depending on the current channel value, thus reducing the complexity burden of the overall approach. Furthermore, we also introduce a sub-optimal strategy based on a recently proposed successive convex approximation (SCA) framework \cite{Scutari_ICASSP14,Scutari_nonconvex}, whose aim is to solve a relaxed version of the original optimization problem. As proved by simulation results, this method makes possible a remarkable reduction of complexity while, at the same time, providing similar performance as the (global) optimal solution.

\begin{figure}[t]
\centering
\includegraphics[width=7cm]{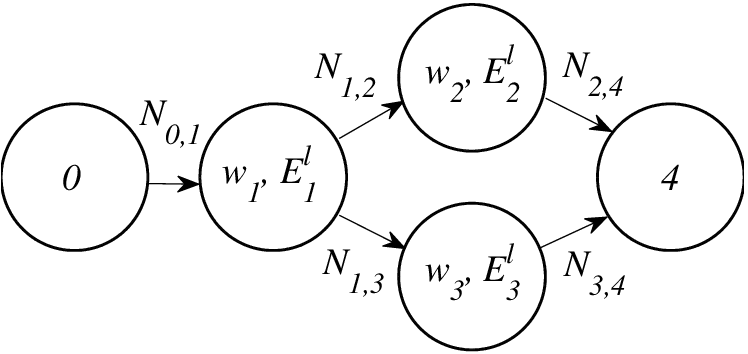}
  \caption{Example of extended call graph.}\label{Call_graph}
\end{figure}

\vspace{-.25cm}
\section{Call Graph}

In this section, before introducing the proposed framework, we recall the notion of call graph of a program \cite{Ryder,Grove-Chambers}, which will be instrumental to design an offloading strategy that decides to offload only some modules of the program, depending on the computation requirements of each module, the size of the program state that needs to be exchanged to transfer the execution from one site to the other, and the channel conditions.

A \emph{call graph} is a useful representation that models the relations between procedures of a computer program into the form of a directed graph ${\cal G}=(\mathcal{V},\mathcal{E})$. The call graph represents the call stack as the program executes. Each vertex $v\in \mathcal{V}$ represents a procedure in the call stack, and each directed edge $e_{u,v}=(u,v)$ represents an invocation of procedure $v$ from procedure $u$. The call graph contains all the relationships among the procedures in a program and, in general, it includes auxiliary information concerning for instance the number of CPU cycles required for running each procedure and the global data shared among procedures. For non recursive languages with reasonable assumptions on the program structure \cite{Ryder}, the call graph is a directed, acyclic graph. In our computation offloading framework, we label each vertex $v\in \mathcal{V}$ with the energy  $E^l_v$ it takes to execute the procedure locally, and with the overall number of CPU cycles $w_v$ associated with the procedure. At the same time, each edge $e_{u,v}=(u,v)$ is characterized by a label describing the number of bits $N_{u,v}$ representing the size of the program state that needs to be exchanged to transfer the execution from node $u$ to node $v$. In general, some procedures cannot be offloaded, like, e.g., the program modules controlling the user interface or the interaction with I/O devices.
We denote the set of procedures that must be executed locally by $\mathcal{V}_l$. To incorporate in our model the fact that the program initiates and terminates at the mobile handset, we consider an extended vertex set $\mathcal{V}_e=\{0,1,\ldots,{\rm card}(\mathcal{V}),{\rm card}(\mathcal{V})+1\}$, where we have introduced the auxiliary nodes 0 and ${\rm card}(\mathcal{V})+1$. These two additive nodes do not introduce further execution energies, but are instrumental to introduce an extended edge set $\mathcal{E}_e$, which comprises all the edges of the call graph, the edges from the initiating node 0 to the call graph, and the edges from the call graph to the terminating node ${\rm card}(\mathcal{V})+1$. The graph ${\cal G}_e=(\mathcal{V}_e,\mathcal{E}_e)$ will be referred to as the extended call graph of the program. To give an example, in Fig. \ref{Call_graph} we illustrate an extended call graph, composed of five nodes and five edges.

\vspace{-.1cm}
\section{Joint Optimization in the Single Channel Case}

The goal of this section is to formulate an optimization problem aimed at determining which modules of an application's call graph should be executed locally and which ones should be executed remotely. Intuitively speaking, the modules more amenable for offloading are the ones requiring intensive computations and limited exchange of data to transfer the execution from one site to the other. Our goal now is to make this intuition the result of an optimization procedure. To this end, we formulate the offloading decision problem jointly
with the selection of the transmit power and the constellation size used for transmitting the program state necessary to transfer the execution from the mobile handset to the cloud or viceversa. The objective is to minimize the energy consumption at the mobile site, under a power budget constraint, and a latency constraint taking into account the time to transfer the execution and the time necessary to execute the module itself. This constraint is what couples the computation and communication aspects of the problem. We assume that the set of instructions to be executed is available at both MUE and SCceNB. If not, they are supposed to be downloaded by the SCceNB through a high speed wired link.

Let us indicate with $\{I_v\}_{v=0}^{{\rm card}(\mathcal{V})+1}$ the indicator variable, which is equal to one, if the procedure $v$ of the extended call graph is executed remotely, or zero, if it is executed locally. Of course, by the definition of the extended call graph, we have $I_0=I_{{\rm card}(\mathcal{V})+1}=0$. We also denote by $p_{u,v}$ the power spent to transmit the program state between the procedures $u$ and $v$. The set of all powers is collected in the vector $\bp=\{p_{u,v}\}_{(u,v)\in \mathcal{E}_e}\in \mathbb{R}^{{\rm card}(\mathcal{E}_e)}$, where ${\rm card}(\mathcal{E}_e)$ is the cardinality of set $\mathcal{E}_e$.
Thus, the optimal allocation of radio parameters and code partitioning can be obtained as the solution of the following optimization problem:
\begin{align}
&[{P.1}]\hspace{.2cm}\min_{\mathbf{I},\bp}  \quad \sum_{v\in \mathcal{V}} (1-I_v)\cdot E_v^l + \sum_{(u,v)\in \mathcal{E}_e} \big[J_{u,v}(p_{u,v})I_v  \nonumber\\
&\hspace{1.5cm}+\varepsilon_{u,v}I_u-(J_{u,v}(p_{u,v})+\varepsilon_{u,v})I_uI_v\big]\label{Integer_Linear_Program}\\
&\hbox{subject to}\nonumber\\
&\; \sum_{v\in \mathcal{V}}\left[(1-I_v)T_v^l+I_vT_v^r\right]+ \sum_{(u,v)\in \mathcal{E}_e} \big[D_{u,v}(p_{u,v})I_v \nonumber\\
&\hspace{1cm}+\gamma_{u,v}I_u-(D_{u,v}(p_{u,v})+\gamma_{u,v})I_uI_v\big] \leq L \label{Integer_Linear_Program2}\\
&\quad \quad  I_v\in\{0,1\}, \quad \quad I_v=0, \quad \forall v\in \mathcal{V}_l, \label{Integer_Linear_Program3}\\
&\quad \quad \; 0\leq p_{u,v} \leq P_T \cdot I_v(1-I_u) \quad \forall\; (u,v)\in \mathcal{E}_e,\label{Integer_Linear_Program4}
\end{align}
where $\mathbf{I}=\{I_{v}\}_{v=0}^{{\rm card}(\mathcal{V})+1}$. The objective function in (\ref{Integer_Linear_Program}) represents the total energy spent by the MUE for executing the application. In particular, the first term in (\ref{Integer_Linear_Program}) is the sum (over all the vertices of the call graph, i.e. $v\in \mathcal{V}$) of the energies spent for executing the procedures locally, whereas the second term is the sum (over the edges of the extended call graph, i.e. $(u,v)\in \mathcal{E}_e$) of the energies spent to transfer the execution from the MUE to the SCceNB, or viceversa. In particular, the quantity $J_{u,v}(p_{u,v})$ represents the energy needed to transmit the program state $N_{u,v}$ from the MUE to the SCceNB, whereas $\varepsilon_{u,v}$ is the cost needed by the MUE to decode the $N_{u,v}$ bits of the program state transmitted back by the SCceNB. The cost $\varepsilon_{u,v}$ is not a function of the MUE's transmitted power and it depends only on the size of the program state $N_{u,v}$. Note that an energy cost to transfer the execution from the MUE to the SCceNB occurs in (\ref{Integer_Linear_Program}) only if the two procedures $u$ and $v$ are executed at different locations, i.e. $I_u\neq I_v$. More specifically, if $I_u=0$ and $I_v=1$, the energy cost is equal to the energy $J_{u,v}(p_{u,v})$ needed to transmit the program state $N_{u,v}$ from the MUE to the SCceNB, whereas, if $I_u=1$ and $I_v=0$, the cost is equal to the energy $\varepsilon_{u,v}$ needed by the MUE to decode the $N_{u,v}$ bits of the program state transmitted by the SCceNB. Now, assuming an adaptive modulation scheme that selects the QAM constellation size as a function of channel conditions and computational requirements, the minimum time $D_{u,v}(p_{u,v})$ necessary to transmit $N_{u,v}$ bits of duration $T_b$ over an additive white Gaussian noise (AWGN) channel is
\begin{equation}\label{Delta_tx}
D_{u,v}(p_{u,v})=\frac{N_{u,v} T_b}{\displaystyle\log_2\left(1+a p_{u,v}\right)},
\end{equation}
where the denominator represents the number of bits/symbol, $p_{u,v}$ is the transmit power, and $a$ is the normalized channel coefficient given by:
\begin{equation}\label{norm_channel_coeff}
a=\frac{\alpha^2\cdot PL(d)}{\Gamma(\mathrm{BER}) N_0},
\end{equation}
where $PL(d)$ denotes the path-loss ($d$ is the distance between MUE and SCceNB), $\alpha^2$ is the channel fading coefficient (normalized to distance), $N_0$ denotes the noise power, and $\displaystyle \Gamma(\mathrm{BER})$ represents the so called SNR margin, introduced to meet the desired target bit error rate (${\rm BER}$) with a QAM constellation.  Expression (\ref{Delta_tx}) shows that offloading may be advantageous if the distance $d$ between MUE and SCceNB is sufficiently small, so that the transmission rate is high.
Thus, exploiting (\ref{Delta_tx}), the energy $J_{u,v}(p_{u,v})$, associated to the transfer of the program state $N_{u,v}$, is given by
\begin{equation}\label{Energy1}
J_{u,v}(p_{u,v})=p_{u,v}D_{u,v}(p_{u,v}).
\end{equation}
The constraint in (\ref{Integer_Linear_Program2}) is a latency constraint and it contains two terms: the first term is the time needed to compute the procedures locally (i.e. $I_v=0$), or remotely (i.e. $I_v=1$), and the second term is the delay resulting from transferring the program state from one site (e.g., the MUE) to the other (e.g., the SCceNB). The constant $L$ represents the maximum latency dictated from the application. In particular, we denote by $T^l_v$ and  $T^r_v$ the time it takes to execute the program module locally or remotely, respectively. They are directly related to $w_v$ through the following expressions $T^l_v=w_v/f_l$, and $T^r_v=w_v/f_s$, where $f_l$ and $f_s$ are the CPU clock speeds (cycles/seconds) at the MUE and at the SCceNB, respectively.
The quantity $D_{u,v}(p_{u,v})$, given by (\ref{Delta_tx}), represents the delay needed to transmit the program state $N_{u,v}$ from the MUE to the SCceNB, whereas $\gamma_{u,v}$ is the time needed by the MUE to decode the $N_{u,v}$ bits of the program state transmitted back by the SCceNB. From (\ref{Integer_Linear_Program2}), we note that no delay occurs if the two procedures $u$ and $v$ are both executed in the same location, i.e. $I_u=I_v$. Furthermore, if $I_u=0$ and $I_v=1$, the delay is equal to the time $D_{u,v}(p_{u,v})$ in (\ref{Delta_tx}) needed to transmit the program state $N_{u,v}$ from the MUE to the SCceNB,  whereas, if $I_u=1$ and $I_v=0$, the delay is equal to the time $\gamma_{u,v}$ needed by the MUE to get the $N_{u,v}$ status bits back from the SCceNB. The term $\gamma_{u,v}$ is not a function of the MUE's transmitted power and it depends only on the size of the program state $N_{u,v}$ to be transferred. The constraint in (\ref{Integer_Linear_Program3}) specifies that the variables $I_v$ are binary and that, for all procedures contained in the set $V_l$, which is the set of procedures that are to be executed locally, $I_v=0$. The last constraint, in (\ref{Integer_Linear_Program4}), is a power budget on the maximum transmit power, which is equal to $P_T$ if $I_u=0$ and $I_v=1$, and zero in all other cases.\smallskip

\textbf{Remark 1:} Since the state variables $I_v$ are integer, $[P.1]$ is a mixed integer programming problem \cite{Nemhauser-Wolsey}. Hence, its solution would require to explore all possible combinations in $I$, solve for the power vector $\bp$, and then compare the final values of the objective function in order to choose the best configuration. However, the number of combinations grows exponentially with the number of vertexes that can be offloaded in the call graph. Furthermore, given a call graph's combination, the resulting optimization problem in $\bp$ might be nonconvex, i.e. multiple local minima may exist. Nevertheless, a series of simplifications are possible in order to reduce the complexity of the overall algorithm. The first useful remark is that, with respect to the integer variables $I_v$, the problem is a {\it quadratic} (binary) integer programming problem. This remark is useful because there exist efficient algorithms, such as the branch and bound algorithm for example, to solve a quadratic integer problem with reduced complexity \cite{Nemhauser-Wolsey}. A further important simplification comes from observing that, for any set of integer values $I_v$, the remaining optimization over the power coefficients is a (\textit{strictly}) \textit{convex} problem. This means that its solution is unique and it can be achieved with very efficient numerical algorithms. Of course, the overall complexity depends on the granularity of the call graph construction: A fine-grain model provides better performance, but with a high cost; conversely, a coarse-grain call graph provides worse performance, but with lower complexity. \qedsymbol

\subsection{Joint Optimization Algorithm}

Before introducing the proposed algorithm, we prove the previous statements regarding the convexity of the power allocation problem. Let us consider a generic combination $c\in \mathcal{C}$, where $\mathcal{C}$ is the set of all possible combinations of the binary variables $I_v$, $v\in \mathcal{V}_e$. For each combination $c$, the value of the variables $I_v$ is fixed to some value $I^c_v$, and problem $[P.1]$ becomes a radio resource allocation problem where we are attempting to minimize the energy spent for transmission with a constraint on the transmission delay, i.e.,
\begin{align} \label{Opt_problem2}
\hspace{-.5cm}[{P.2}]\hspace{.5cm} \min_{\displaystyle \bp_c} & \quad \sum_{(u,v)\in \mathcal{E}_{e,c}}\frac{N'_{u,v}p_{u,v}}{\displaystyle\log\left(1+a p_{u,v}\right)} \nonumber\\
& \hbox{s.t.} \quad \sum_{(u,v)\in \mathcal{E}_{e,c}} \frac{N'_{u,v}}{\displaystyle\log\left(1+a p_{u,v}\right)}\leq L_c \nonumber\\
& \quad \quad \; 0\leq p_{u,v} \leq P_T, \quad\quad \forall\; (u,v)\in \mathcal{E}_{e,c}, \nonumber
\end{align}
where, for any combination $c$,  $\mathcal{E}_{e,c}$ is the set of edges $(u,v)\in \mathcal{E}_e$ for which $I^c_u=0$ and $I^c_v=1$;
$\bp_c=\{p_{u,v}\}_{(u,v)\in \mathcal{E}_{e,c}}$ is vector containing all the powers; $N'_{u,v}=N_{u,v}T_b$
and
\begin{equation}\label{L1}
L_c=L-\sum_{v\in \mathcal{V}}\left[(1-I^c_v)\frac{w_v}{f_l}+I^c_v\frac{w_v}{f_s}\right]-\sum_{(u,v)\in \mathcal{E}'_{e,c}}\gamma_{u,v}
\end{equation}
where $\mathcal{E}'_{e,c}$ is the set of edges $(u,v)\in \mathcal{E}_{e}$ for which $I^c_u=1$ and $I^c_v=0$. Going from $[P.1]$ to $[P.2]$, we have eliminated all the terms that do not depend explicitly on the transmit powers, since they do not affect the optimization problem $[P.2]$. The problem in $[P.2]$ is still nonconvex because the objective function is actually concave in the power allocation vector $\bp_c$. Nevertheless, we can prove the following result: \smallskip
\textit{\begin{theorem}
If the following conditions
\begin{equation}\label{feasibility}
L_c>0 \quad\ \hbox{and} \quad a\geq \bar{a}_c=\frac{e^{\displaystyle \frac{\sum_{(u,v)\in \mathcal{E}_{e,c}}N'_{u,v}}{L_c}}-1}{P_T}
\end{equation}
are satisfied, then:
\begin{description}
  \item[(a)] the feasible set of problem $[P.2]$ is non-empty;
  \item[(b)] Problem $[P.2]$ is equivalent to the convex optimization problem $[P.7]$ (see Appendix A);
\end{description}
\end{theorem}
\begin{proof}
See Appendix A. \smallskip
\end{proof}}

\textbf{Remark 2:} By virtue of Theorem 1 and the results in Appendix A, the delay constrained energy minimization problem in $[P.2]$ is equivalent to a convex problem with strictly convex objective function, see $[P.7]$. Thus, if the feasible set is non-empty, the problem admits a unique global solution that can be found using numerically efficient algorithms \cite{Boyd}. This is a very important result, since it enables the possibility to find the global optimum solution of $[P.1]$ through a search on the \textit{useful} call graph configurations. Indeed, the feasibility conditions (\ref{feasibility}) enables to discard several call graph configurations that cannot afford offloading, thus reducing the complexity burden of the algorithm. The computation of the useful graph configurations is very easy. Indeed, since $L_c$ in (\ref{L1}) and the threshold values $\bar{a}_c$ in  (\ref{feasibility}) are fixed coefficients depending on the call graph structure, the delay constraint $L$, and the power budget constraint $P_T$, these parameters must not be computed on the fly but they can be precomputed and stored in the memory. Thus, when computation offloading is requested, the selection of the useful graph configuration can be done by simply checking the conditions in (\ref{feasibility}), given the current equivalent channel $a$ between the MUE and the SCceNB. The search will then be limited only to the useful call graph configurations.\qedsymbol

\textbf{Remark 3:} The problem formulation in $[P.1]$ extends the MAUI strategy from \cite{Maui}, where it was proposed an integer linear program (ILP) aimed at finding the optimal program partitioning strategy that minimizes the MUE's energy consumption, subject to a latency constraint. Differently from our work, in \cite{Maui} the quantities $D_{u,v}(p_{u,v})$ and $J_{u,v}(p_{u,v})$ were assumed to be constant (and known), whereas in our case those quantities are optimized jointly with the program partition, acting on the transmit power and the constellation size. This introduces the radio aspect in the problem, and leads to a further gain in terms of energy saving that can be achieved through computation offloading. From a complexity point of view, our approach and MAUI are both NP-hard problems. Nevertheless, thanks to the check of the feasibility conditions in (\ref{feasibility}), which depend on the radio aspect of the problem, several graph configurations can be apriori discarded depending on the current channel value. Thus, our approach might largely reduce the overall complexity burden of computation offloading, with respect to the MAUI strategy. \qedsymbol

The main steps of the algorithm that solves the joint optimization of radio resources and code partitioning are summarized in Algorithm 1.

\begin{algorithm}[t]

\begin{flushleft}
\justify

\normalsize

\vspace{.15cm}
(\texttt{S.1}): Given the equivalent channel $a$, select the set of combinations $\mathcal{C}$ that satisfy the feasibility conditions in (\ref{feasibility});

\vspace{.15cm}
\noindent (\texttt{S.2}): For every combination $\bI_c=\{I_v^c\}_{v=0}^{{\rm card}(\mathcal{V})+1}$, $c\in\mathcal{C}$, compute the optimal power vector $\bp^*_c$ by solving the convex problem $[P.7]$ (Appendix A). Then, using (\ref{Integer_Linear_Program}), evaluate the overall energy $E_c(\bI_c,\bp^*_c)$ spent to run the application, for the given combination $c$;

\vspace{.15cm}
\noindent (\texttt{S.3}):  Select the pair $\displaystyle (\bI^*_c,\bp^*_c)=\arg\min_{c\in \mathcal{C}} E_c(\bI_c,\bp^*_c)$.

\caption{\textbf{:} Joint Optimization Algorithm }

\end{flushleft}

\end{algorithm}

\vspace{-.2cm}
\subsection{Numerical Results}

In this section, we provide some numerical results to assess the tradeoff between performance and complexity of the proposed offloading strategy.

\noindent{\it Numerical Example 1 - Performance:} The performance of the proposed offloading strategy  is affected by several factors. The first important factor is the call graph topology and its parameters. As an example, let us consider three different examples of call graphs, as illustrated in Fig. \ref{Graph}. To grasp the dependence of the performance on the graph's parameters, we consider for simplicity three directed call graphs having the same topology, but different parameters. In Fig. \ref{Graph}, nodes 0 and 7 represent the auxiliary nodes of the extended call graph, and are useful to model the user interface. All other nodes of the graphs (from 1 to 6) can be offloaded, if needed. All the nodes in the three graphs have the same set of energies $E^l_v$ (Joule) and number of CPU cycles $w_v$ (Mcycles), associated to the local computation of a method. In particular, $E^l_1=4.2$, $E^l_2=5.1$, $E^l_3=6.7$, $E^l_4=8.4$, $E^l_5=0.5$, $E^l_6=0.4$, and $w_1=70.8$, $w_2=42.5$, $w_3=95.3$, $w_4=158.6$. $w_5=18.6$, $w_6=86.4$. What is different from one graph to the other is the number $N_{u,v}$ of Kbit associated to the program state at each edge of the extended call graph, as we can see in Fig. \ref{Graph}.

\begin{figure}[t]
\centering
\includegraphics[width=7cm]{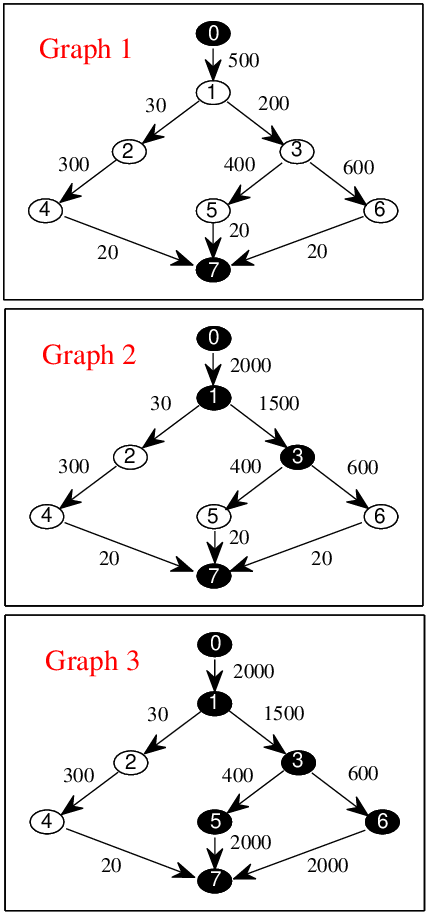}
  \caption{Examples of optimal graph's partitions.}\label{Graph}
\end{figure}

The colors associated to the nodes in Fig. \ref{Graph} denote the final result of the proposed joint optimization algorithm for the specific graph. In particular, the black color is assigned to those nodes for which the optimization strategy has decided to run the corresponding procedure locally, whereas the white color is assigned to the procedures that are offloaded to the SCceNB. The latency constraint $L$ in $[P.1]$ is chosen equal to the time needed by the MUE to compute the entire program locally, and the power budget constraint is $P_T=0.01$ Watt. The local CPU clock speed is chosen as $f_l=10^8$ cycles/s, whereas the server speed is equal to $f_s=10^{10}$ cycles/s. The normalized channel coefficient $a$ in (\ref{norm_channel_coeff}) is set equal to 500, and the bit duration $T_b=10^{-6}$.
As we can notice from Fig. \ref{Graph}, slightly different call graphs (only in terms of number of bits $N_{u,v}$) lead to completely different results. In the first case (Graph 1), all the nodes (except the auxiliary nodes, of course) are offloaded to the server. This happens because the overall energy that the MUE should have spent for computing the entire application locally is much greater than the energy needed to transmit the initial state of size $N_{01}=500$ KByte.  In this case, it might be that, for each individual method after node 2, its remote execution is more expensive than its local execution, and yet remote execution can save energy by offloading the entire program. It is indeed important to remark that the result of the optimization is globally optimal (i.e., across the entire program) rather than locally optimal (i.e., relative to a single method invocation).  Let us consider now the case of Graph 2 in Fig. \ref{Graph}, where the size $N_{01}$ and $N_{13}$ of the program state are larger than before. In this case, the optimal solution is the one shown in Fig. \ref{Graph} (Graph 2), where only four methods are offloaded. This happens because, from an energetic point of view, it has become non convenient to transmit the large states $N_{01}$ or $N_{13}$ to the SCceNB. Furthermore, once a computation is offloaded, the result must be transmitted back to the SCceNB. If the result yields too many bits, it may not be convenient to offload that procedure. This is exactly what happen in the case of Graph 3 in Fig. \ref{Graph}. Indeed, increasing the values of the output size of procedures 5 and 6 with respect to the case of Graph 2, the optimal solution changes as shown in Fig. \ref{Graph} (Graph 3).

\begin{figure}[t]
\centering
\includegraphics[width=7.5cm]{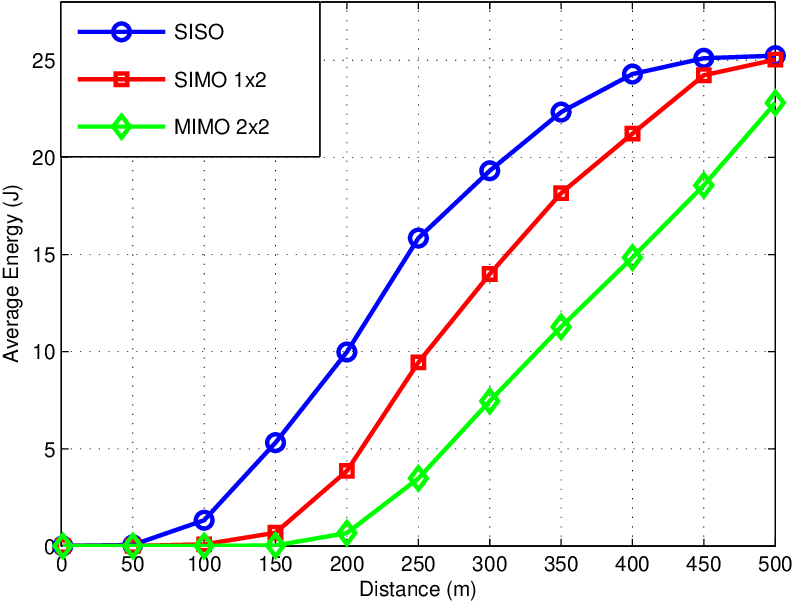}
  \caption{Average energy spent for processing versus distance (from MUE to SCceNB), for different communication strategies.}\label{Energy_distance}
\end{figure}

A fundamental factor affecting the performance of the joint optimization algorithm is the wireless fading channel between MUE and SCceNB. For this reason, we check the performance of the optimization algorithm as a function of the distance between MUE and SCceNB, considering a Multiple-Input-Multiple-Output (MIMO) transmission scheme. In Fig. \ref{Energy_distance}, we report the average energy spent for processing versus the distance between MUE and SCceNB, for different MIMO configurations. The results are averaged over 200 independent realizations.
Assuming the presence of $M$ statistically independent spatial channels, the fading coefficient $\alpha$ is a random variable, whose probability density function (pdf) depends on the MIMO communication strategy, with $\bar{\alpha}$ denoting the variance over the single channel coefficient \cite{Barbarossa}. In particular, in Fig. \ref{Energy_distance}, we assume the use of SISO, 1x2 SIMO, and 2x2 MIMO communication strategies. The BER is chosen equal to $10^{-3}$, the path-loss coefficient $PL(d)$ is chosen according to the small cell model in \cite{3GPP} (with carrier frequency equal to 2 GHz), the noise power is equal to $-135$ dB, and the variance $\bar{\alpha}=1$. Furthermore, we consider the Graph 1 in Fig. \ref{Graph} as a call graph for this example. As expected, from Fig. \ref{Energy_distance}, we notice how computation offloading is more convenient if the distance between MUE and SCceNB is sufficiently low, in order to allow the offloading of the entire program with high probability. This is a further numerical justification for favoring the access to the cloud through small cells. From Fig. \ref{Energy_distance}, we also notice how, increasing the number of antennas, we get a larger energy saving thanks to the increased offloading of computing tasks.

\begin{figure}[t]
\centering
\includegraphics[width=7.5cm]{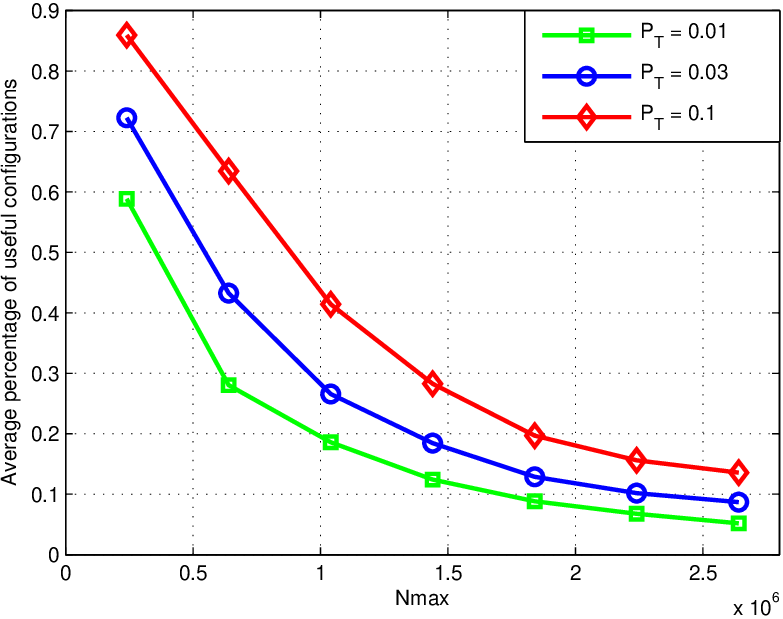}
  \caption{Average percentage of useful call graph configurations versus the maximum size $N_{\max}$ of the program state, for different power budgets $P_T$.}\label{Config_delta}
\end{figure}

\begin{figure}[t]
\centering
\includegraphics[width=7.5cm]{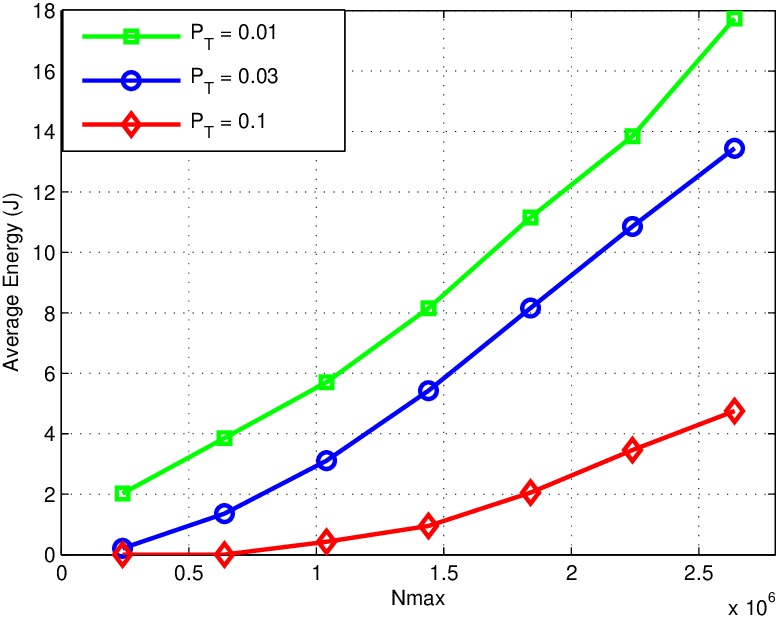}
  \caption{Average energy spent for offloading versus the maximum size $N_{\max}$ of the program state, for different power budgets $P_T$.}\label{Energy_delta}
\end{figure}

\noindent{\it Numerical Example 2 - Energy Saving/Complexity Tradeoff:} The aim of this example is to show the existing tradeoff between the performance achievable by the proposed offloading strategy and its computational complexity. Considering a call graph having $\bar{V}$ offloadable vertexes, in principle the search procedure should compare the results achieved over all the $2^{\bar{V}}$ possible configurations. However, given a certain realization of wireless channels, only a subset of these configurations satisfy the feasibility conditions in (\ref{feasibility}). It is then of interest to investigate on the effective percentage of useful configurations that the proposed algorithm must check with respect to variations of the application's parameters, the radio channel, and the radio requirements on the power budget. The results have been obtained in the following way. Starting from a graph structure as in Fig. \ref{Graph}, we assume a maximum size $N_{\max}$ of the program state to be transferred and a maximum value $w_{\max}$  of the computational load (in terms of CPU cycles). Then, we generate one thousand graphs having the structure of Fig. \ref{Graph}, but with values of $N_{u, v}$ and $w_v$ generated as uniform random variables in the interval $(0, N_{\max}]$ and $(0, w_{\max}]$. Thus, in Figs. \ref{Config_delta} and \ref{Energy_delta}, we illustrate the average behavior of the number of useful call graph configurations that the algorithm must check and the optimal energy spent by the proposed offloading strategy, versus $N_{\max}$, while setting $w_{\max}=10^7$. In Figs. \ref{Config_delta} and \ref{Energy_delta}, we also consider different values of the power budget constraint $P_T$. The channel parameters are chosen as in the previous simulation, and the MUE/SCceNB distance is set equal to 100 meters. The other parameters are the same of the previous example. The result reported in Fig. \ref{Config_delta} shows that, as $N_{\max}$ increases, there are less and less feasible configurations to check, so that the overall computational load greatly diminishes with respect to the search over all the possible call graph's combinations.  We also see how, increasing the transmit power $P_T$, there are more feasible configurations because offloading is more likely to occur. At the same time, in Fig. \ref{Energy_delta}, we report the energy consumption averaged over all call graph realizations, for different values of the power budget constraint $P_T$. From Fig. \ref{Energy_delta}, we can notice how low energy consumptions are obtained as $N_{\max}$ decreases. In this case, the energy consumption is smaller for the higher transmit power, because offloading occurs more frequently. Thus, Figs. \ref{Config_delta} and \ref{Energy_delta} show an interesting tradeoff between complexity and performance achieved by the proposed offloading strategy: lower energy consumptions are achieved at the price of a larger complexity for solving Algorithm 1.

\section{Joint Optimization in the Multiple \\Channel Case}

We now extend the previous offloading optimization strategy to the wideband channel case, where the MUE transmits over a set of $K$ parallel sub-bands. In this case, we have the extra degree of freedom on how to distribute the transmit power across the parallel channels. In the multi-channel case, the average delay associated to the correct reception of $N_{u,v}$ bits at the server side is given by:
\begin{align}\label{Avg_Delay2}
D_{u,v}(\bp_{u,v})=\frac{N_{u,v}'}{\displaystyle\sum_{k=1}^K\log\left(1+a_k p^k_{u,v}\right)}
\end{align}
\begin{equation}
\displaystyle \hbox{with} \quad N_{u,v}'=N_{u,v} T_b\log2, \quad \displaystyle a_k=\frac{\alpha_k^2\cdot PL(d)}{\Gamma(\mathrm{BER}) N_0},
\end{equation}
where $\alpha_k^2$ and $p^k_{u,v}$ denote the fading coefficient and the power transmitted over the $k$-th subchannel, respectively, and  $\bp_{u,v}=[p^1_{u,v},\ldots,p^K_{u,v}]^T$. Exploiting (\ref{Avg_Delay2}), the average energy cost $J_{u,v}(\bp_{u,v})$, associated to the transfer of the program state $N_{u,v}$, is given by
\begin{equation}\label{Energy3}
J_{u,v}(\bp_{u,v})=D_{u,v}(\bp_{u,v})\sum_{k=1}^K p^k_{u,v}.
\end{equation}
Let us define $\bp=\{\bp_{u,v}\}_{(u,v)\in \mathcal{E}_e}\in \mathbb{R}^{K\cdot{\rm card}(\mathcal{E}_e)}$. Exploiting (\ref{Avg_Delay2})-(\ref{Energy3}), the joint optimization of radio resources and code partitioning, in the case of transmission over a wideband channel, can be cast as:
\begin{align} \label{Opt_problem4}
&[{P.3}]\hspace{.2cm}\min_{\mathbf{I},\mathbf{p}\small{\geq}\mathbf{0}}  \quad \sum_{v\in \mathcal{V}} (1-I_v)\cdot E_v^l + \sum_{(u,v)\in \mathcal{E}_e} \big[J_{u,v}(\bp_{u,v})I_v  \nonumber\\
&\hspace{1.5cm}+\varepsilon_{u,v}I_u-(J_{u,v}(\bp_{u,v})+\varepsilon_{u,v})I_uI_v\big] \nonumber\\
&\hbox{subject to}\nonumber\\
& \; \sum_{v\in \mathcal{V}}\left[(1-I_v)T_v^l+I_vT_v^r\right]+ \sum_{(u,v)\in \mathcal{E}_e} \big[D_{u,v}(\bp_{u,v})I_v \nonumber\\
&\hspace{1cm}+\gamma_{u,v}I_u-(D_{u,v}(\bp_{u,v})+\gamma_{u,v})I_uI_v\big] \leq L \nonumber\\
& \quad \quad  I_v\in\{0,1\}, \quad \quad I_v=0, \quad \forall v\in \mathcal{V}_l, \nonumber \\
& \quad \;\; \; \sum_{k=1}^K p^k_{u,v} \leq P_T\cdot I_v(1-I_u), \quad \forall\; (u,v)\in \mathcal{E}_e, \nonumber
\end{align}
where $J_{u,v}(\bp_{u,v})$ and $D_{u,v}(\bp_{u,v})$ are given by (\ref{Energy3}) and (\ref{Avg_Delay2}), respectively. The last constraint in $[P.3]$ denotes the budget on the sum of transmit powers, which is equal to $P_T$ if $I_u=0$ and $I_v=1$, and zero in all other cases.

\vspace{-.3cm}
\subsection{Joint Optimization Algorithm}

Proceeding as for $[P.1]$, let us consider a generic combination $c\in \mathcal{C}$, where $\mathcal{C}$ is the set of all possible combinations of the binary variables $I_v$, $v\in \mathcal{V}$. For each combination $c$, the value of the variables $I_v$ is fixed to some value $I^c_v$, and problem $[P.3]$ becomes the radio resource allocation problem
\begin{align}
\hspace{-.2cm}[{P.4}]\hspace{.5cm}\min_{\mathbf{p}_c\geq\mathbf{0}} & \quad \sum_{(u,v)\in \mathcal{E}^c_e}  \frac{\displaystyle N_{u,v}'\sum_{k=1}^K p^k_{u,v}}{\displaystyle\sum_{k=1}^K\log\left(1+a_k p^k_{u,v}\right)}   \nonumber\\
& \hbox{s.t.} \quad  \sum_{(u,v)\in \mathcal{E}^c_e} \frac{N_{u,v}'}{\displaystyle\sum_{k=1}^K\log\left(1+a_k p^k_{u,v}\right)}   \leq L_c \nonumber\\
& \;\; \;\quad\quad \sum_{k=1}^K p^k_{u,v} \leq P_T, \quad \forall\; (u,v)\in \mathcal{E}_e^c, \nonumber
\end{align}
where $\mathcal{E}^c_e$ is the set of edges $(u,v)\in \mathcal{E}_e$ for which $I^c_u=0$ and $I^c_v=1$ given the combination $c$. Furthermore, the optimization vector is $\bp_c=\{\bp_{u,v}\}_{(u,v)\in \mathcal{E}_e^c}$, where $\bp_{u,v}=[p^1_{u,v},\ldots,p^K_{u,v}]^T$. The delay constraint $L_c$ is the same as (\ref{L1}). Again, the problem in $[P.4]$ is nonconvex. However, it can be proved the following result:

\vspace{-.3cm}
\textit{\begin{theorem}
If the following conditions are satisfied:
\begin{equation}\label{feasibility2}
L_c>0 \quad \hbox{and} \quad \sum_{k=1}^K \log(1+a_k p^*_k)\geq \frac{\sum_{(u,v)\in \mathcal{E}^c_e} N'_{u,v}}{L_c},
\end{equation}
where $\bp^*=\{p^*_k\}_{k=1}^K$ is the solution of the water-filling problem (\ref{WF}) (see Appendix B), then:
\begin{description}
  \item[(a)] problem $[P.4]$ admits a non-empty feasible set;
  \item[(b)] problem $[P.4]$ is equivalent to the problem $[P.8]$ (see Appendix B), which is such that any local optimum is also globally optimal.
  \end{description}
\end{theorem}
\begin{proof}
See Appendix B.
\end{proof}}

To solve the problem $[P.3]$, we propose an algorithm totally similar to Algorithm 1. The only difference is that, in (S.1), to find the set of useful combinations $\mathcal{C}$, we check the feasibility condition (\ref{feasibility2}), instead of (\ref{feasibility}), and in (S.2) we solve the optimization problem in $[P.4]$ $\left(\hbox{i.e., }[P.8]\right)$, instead of $[P.2]$ $\left(\hbox{i.e., }[P.7]\right)$.

\textbf{Remark 4:} To check the feasibility condition in (\ref{feasibility2}), the solution of the water-filling problem (\ref{WF}) (see Appendix B) must be first computed. This solution is given by a simple iterative algorithm that converges in (at most) $K$ iterations. Thus, since the complexity of the Water-Filling computation is very low, the selection of the useful call graph configurations is very efficient also in the wide-band channel case.\qedsymbol

\subsection{Numerical Results}

We now apply our proposed joint optimization approach to the case of a realistic call graph of a program, representing a face recognition application, see \cite{Maui}. The application's call graph is the same shown in Fig. \ref{Call_graph}, where nodes 0 and 4 represent the user interface and cannot be offloaded. The other parameters are
$E_1^l=0.872$, $E_2^l=4.703$, $E_3^l=13.03$ (Joule), $w_1=18.1$, $w_2=92.6$, $w_3=256.1$ (Mcycles), and $N_{0,1}=182$, $N_{1,2}=4675$, $N_{1,3}=13860$ (KByte). For this application, we illustrate an example of optimal power allocation and call graph's partition, in the case the MUE is transmitting by using multiple channels. For this purpose, in Fig. \ref{FR_MC_Optimization}, we report the result of our joint optimization algorithm. We consider an OFDM system with $K=8$ subcarriers.  For the problem in $[P.3]$, the latency constraint $L$ is chosen equal to the time needed by the MUE to compute the entire program locally, and the power budget constraint is $P_T=0.018$ Watt. The local CPU clock speeds $f_l$ and $f_s$, and the bit duration $T_b$ are chosen as in previous simulations. The normalized wireless channels $a_k$ between MUE and SCceNB are given in Fig. \ref{FR_MC_Optimization} (middle). The optimal graph's partition for this parameter setting is shown in Fig. \ref{FR_MC_Optimization} (bottom), where, again, the white nodes denote procedures computed at the SCceNB side. As we can notice from Fig. \ref{FR_MC_Optimization} (bottom), all the remoteable nodes are offloaded to the SCceNB. This means that the optimization has found that the most convenient solution in terms of energy is to transmit the 182 KB of the program state between the nodes 0 and 1, and then compute all the rest of the program at the server. Thus, only one link in the call graph is selected to be used for the transmission of data. In particular, the top plot in Fig. \ref{FR_MC_Optimization} (top) shows the optimal power allocation over the multiple channels, achieved as a result of the optimization problem $[P.4]$. As we can see from Fig. \ref{FR_MC_Optimization} (top and middle), the power allocation shows a water-filling behavior, where all the power is concentrated over the best channels, while no bits are transmitted over the worse channels. The energy saving is potentially huge in this case. Indeed, by computing locally nodes 1, 2, and 3 of the call graph, the MUE would have spent 18.6 Joule, while, by offloading the entire program as in Fig. \ref{FR_MC_Optimization} (bottom), the MUE would spend only 25 mJoule. The gain in terms of energy saving is about 740 times.

\begin{figure}[t]
\centering
\includegraphics[width=8cm]{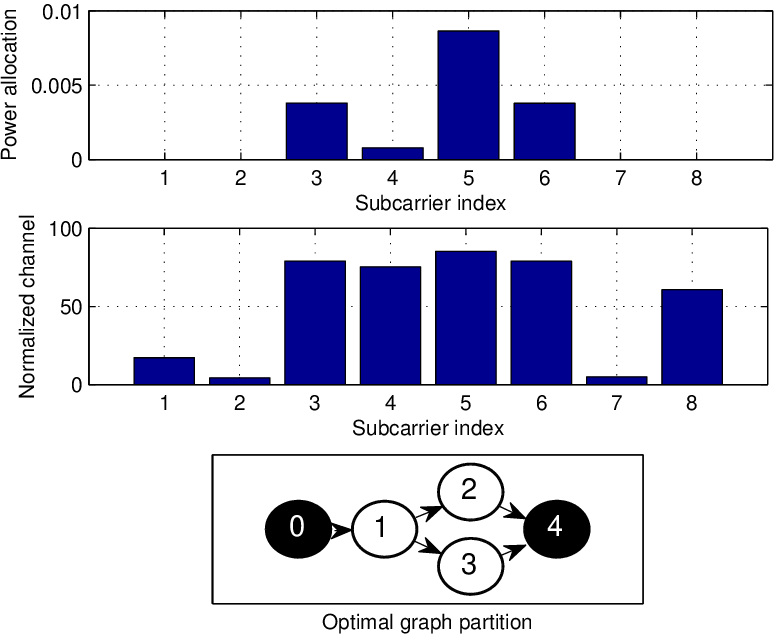}
  \caption{Result of the Joint optimization. (Top) Optimal power allocation. (Middle) Normalized channel coefficients. (Bottom) Optimal graph's partition.}\label{FR_MC_Optimization}
\end{figure}

\section{Relaxed Optimization based on Successive Convex Approximation}

As discussed above, finding the optimal solution of the mixed integer programming $[P.1]$ requires an exhaustive search over all the useful call graph partitions, i.e. all the configurations that satisfy the feasibility conditions in (\ref{feasibility}). Nevertheless, even if the check of conditions (\ref{feasibility}) might largely reduce the complexity of the overall search, it is clear that such a solution can be adopted only in the case of call graphs with a moderate number of nodes. To face with large size call graphs, we propose now a low-complexity algorithm that is still able to obtain high-quality solutions of problem $[P.1]$ (the same approach can be used to solve also problem $[P.3]$). By hinging on recent successive convex approximation techniques \cite{Scutari_ICASSP14},\cite{Scutari_nonconvex}, we devise an iterative algorithm where the nonconvex optimization problem $[P.1]$ is replaced by a sequence of strongly convex problems. To this end, we relax the original problem by converting the integer variables $\{I_v\}$ into real variables belonging to the interval $[0,1]$. To simplify our formulation we introduce in $[{P.5}]$ the change of variables  $t_{u,v}=I_v(1-I_u)$ where
$t_{u,v}\in [0,1]$, $\forall u,v \in \mathcal{E}_e$. Hence, defining $\mathbf t\triangleq(t_{u,v})_{\forall u,v \in \mathcal{E}_e}$, $[{P.1}]$ can be relaxed into the following nonconvex problem:
\begin{align}
&[{P.5}]\hspace{.2cm}\min_{\bI,\bt,\bp}  \!\quad E(\bI,\bt,\bp)\triangleq\sum_{v\in \mathcal{V}} (1-I_v)\cdot E_v^l  \nonumber\\
&\qquad+\sum_{(u,v)\in \mathcal{E}_e} \!\! \!\big[J_{u,v}(p_{u,v})\cdot t_{u,v} \nonumber +\varepsilon_{u,v}(t_{u,v}-I_v+I_u)\big]\label{Integer_Linear_Program}
\end{align}
\begin{align}
&\hbox{subject to}\nonumber\\
&\hspace{.15cm}\left.\begin{array}{llll}  \texttt{a)}\;\;  \bar{D}(\bI,\bt,\bp)  \leq L, \medskip\\
\texttt{b)}\;\;  I_v(1-I_u) - t_{u,v}\leq 0, \quad \;\, \forall \;u,v\in \mathcal{E}_e,\medskip\\
\texttt{c)}\;\;   t_{u,v}\in [0,1],  I_v\in [0,1], \quad  \forall \;u,v\in \mathcal{E}_e,\medskip\\
 \texttt{d)}\;\;  I_v=0, \hspace{2.7cm}  \forall \; v\in \mathcal{V}_l,\medskip\\
 \texttt{e)}\;\;  0\leq p_{u,v} \leq P_T \cdot t_{u,v},  \quad\quad\forall\; u,v\in \mathcal{E}_e,\medskip\\
  \end{array}\hspace{-0.1cm} \right\}\!\! \triangleq \mathcal{X} \nonumber
\end{align}
\begin{align}
\hbox{where}\quad&\bar{D}(\bI,\bt,\bp)\triangleq\displaystyle \sum_{v\in \mathcal{V}}\left[(1-I_v)T_v^l+I_vT_v^r\right]\nonumber\\
&\hspace{-.3cm}+\displaystyle\sum_{(u,v)\in \mathcal{E}_e} \big[D_{u,v}(p_{u,v})t_{u,v} +\gamma_{u,v}(t_{u,v}-I_v+I_u)\big]. \nonumber
\end{align}
Let $\by\triangleq (\bI,\bt,\bp)$ and $\by^{\nu}\triangleq (\bI^{\nu},\bt^{\nu},\bp^{\nu})$,  with $\nu$ denoting
the temporal index of the iterative procedure that aims to solve $[P.5]$. To find a suitable convexification of $[P.5]$, we need to define valid surrogates for the nonconvex objective function $E(\bI,\bt,\bp)$, and the constraints in a) and b). Proceeding as in \cite{Scutari_ICASSP14}, \cite{Scutari_nonconvex}, we first build a (strongly convex) surrogate function $\tilde{E}(\by;\by^{\nu})$ that approximates the nonconvex objective function $E(\by)$ around the current iterate $\by^{\nu}\in \mathcal{X}$ and enjoys some desirable properties (see \cite{Scutari_ICASSP14}, \cite{Scutari_nonconvex} for further details).
To build $\tilde{E}$, let us preserve the convex terms in $E(\bI,\bt,\bp)$ while convexifying the non-convex part
$\sum_{(u,v)\in \mathcal{E}_e} J_{u,v}(p_{u,v})\cdot t_{u,v}$.
More formally, let us introduce the strongly convex function:
\begin{align}
\tilde{E}(\by;\by^{\nu})& \triangleq \sum_{v\in \mathcal{V}} (1-I_v)\cdot E_v^l + \hspace{-.2cm}\sum_{(u,v)\in \mathcal{E}_e} \hspace{-.2cm}\varepsilon_{u,v}(t_{u,v}-I_v+I_u)\nonumber\\
&\hspace{-.5cm} +\sum_{(u,v)\in \mathcal{E}_e} \left(p^{\nu}_{u,v} D_{u,v}(p_{u,v})+p_{u,v} D_{u,v}(p^{\nu}_{u,v})\right)\cdot {t^{\nu}_{u,v}}\nonumber\\
&\hspace{-.5cm} + \sum_{(u,v)\in \mathcal{E}_e} p^{\nu}_{u,v} D_{u,v}(p_{u,v}^{\nu})\cdot {t_{u,v}}
 + \tau_p \parallel \bp-\bp^{\nu} \parallel^2 \nonumber\\
&\hspace{-.5cm} +\tau_I \parallel \bI-\bI^{\nu} \parallel^2+\tau_t \parallel \bt-\bt^{\nu} \parallel^2 \nonumber
\end{align}
where $\tau_p$, $\tau_I$, $\tau_t$ are positive constants. Note that $\tilde{E}(\by;\by^{\nu})$ has some desirable properties (see \cite{Scutari_ICASSP14},\cite{Scutari_nonconvex}) such as: i) it is uniformly strongly convex with Lipschitz continuous gradient on $\mathcal{X}$; ii) it preserves the same first-order behavior of the original nonconvex function $E(\by)$ at $\by^{\nu}$, i.e., $\nabla_{\by} \tilde{E}(\by^{\nu};\by^{\nu})=\nabla_{\by}{E}(\by^{\nu})$,  $\forall\by^{\nu}\in \mathcal X$ (see \cite{Scutari_ICASSP14},\cite{Scutari_nonconvex}). Now, similarly to what we did for the objective function, we introduce an inner convexification of the (nonconvex) constraints a) and b) in $[{P.5}]$. Let us first build a convex approximation,
say $\tilde{D}(\by;\by^{\nu})$, of the delay constraint $\bar{D}(\by)$ around the current iterate
$\nu$. To do so, we convexify the term  $D_{u,v}(p_{u,v}) t_{u,v}$, for each $u,v \in \mathcal{E}_e$, so that an inner convexification of the constraint a) in $[{P.5}]$ is given by:
\beq
\begin{split}
\tilde{D}(\by;& \by^{\nu}) = \displaystyle \sum_{v\in \mathcal{V}}\left[(1-I_v)T_v^l+I_v T_v^r\right]+ \hspace{-.3cm}\displaystyle\sum_{(u,v)\in \mathcal{E}_e} \big[D_{u,v}(p^{\nu}_{u,v})  \\ &  \hspace{-.5cm}\cdot  \frac{t_{u,v}}{2}+D_{u,v}(p_{u,v})\frac{t^{\nu}_{u,v}}{2} +\gamma_{u,v}(t_{u,v}-I_v+I_u)\big]\leq L.\nonumber
\end{split}
\eeq
Finally we introduce the following inner convexification $\tilde{g}_{u,v}(\by;\by^{\nu})$ of the constraint b) in  $[{P.5}]$:
\beq
\tilde{g}_{u,v}(\by;\by^{\nu})\triangleq \ds \frac{I_v}{2}(1-I^{\nu}_u)+\ds \frac{I^{\nu}_v}{2}(1-I_u) - t_{u,v}. \nonumber
\eeq
The approximations $\tilde{D}(\by; \by^{\nu})$
and $\tilde{g}_{u,v}(\by;\by^{\nu})$ are chosen in order to satisfy some key properties, see \cite{Scutari_ICASSP14},\cite{Scutari_nonconvex} for details.

\subsection{Inner SCA algorithm}

Exploiting the previous arguments, we are now able to introduce the proposed (strongly) convex approximation of problem $[{P.5}]$
around a feasible point $\by^{\nu}$. In particular, replacing the nonconvex objective $E(\by)$ and the nonconvex constraints in a) and b) with the approximations $\tilde{E}(\by;\by^{\nu})$, $\tilde{D}(\by;\by^{\nu})$ and $\tilde{g}(\by;\by^{\nu})$, respectively, we have
\begin{align}
[{P.6}] &\quad  \underset{\mathbf{\by}}{ \text{min}} \quad \tilde{E}(\by;\by^{\nu}) \nonumber\label{Integer_Linear_Program}\\
&\quad  \hbox{subject to}\nonumber\\
&\quad   \hspace{0.15cm}\begin{array}{llll}  \texttt{a)}\;\; \tilde{D}(\by;\by^{\nu})  \leq L \medskip\\
\texttt{b)}\;\;  \tilde{g}_{u,v}(\by;\by^{\nu}) \leq 0,\quad \quad \quad \quad \forall u,v\in \mathcal{E}_e,\medskip\\
\texttt{c)}\;\;   t_{u,v}\in [0,1],  I_v\in [0,1], \quad  \forall \;u,v\in \mathcal{E}_e,\medskip\\
 \texttt{d)}\;\;  I_v=0, \hspace{2.7cm}  \forall \; v\in \mathcal{V}_l,\medskip\\
 \texttt{e)}\;\;  \delta^\nu \leq p_{u,v} \leq P_T \cdot t_{u,v}+\delta^\nu,  \quad\forall\; u,v\in \mathcal{E}_e.
\end{array}\nonumber
\end{align}
In constraint e), we introduced the sequence $\delta^\nu$ that is instrumental to preserve the feasibility of $\mathcal{X}$ while driving the transmission powers to vanishing values. To guarantee that the optimal solutions of $[P.6]$ coincide with those of $[P.5]$, it must hold $\lim_{\nu\rightarrow\infty}\delta^\nu=0$. Then, starting from a feasible point $\by^{0}\triangleq({\bI}^{0}, {\bt}^{0},{\bp}^{0})$ (satisfying the feasibility conditions in (\ref{feasibility})), the proposed method consists in solving the sequence of problems $[{P.6}]$ by following the steps described in Algorithm 2. Note that, in step (S.3), the algorithm exploits a diminishing step-size sequence $\beta^\nu$, such that $\sum_{\nu=0}^\infty \beta^\nu=\infty$, and $\sum_{\nu=0}^\infty (\beta_\nu)^2<\infty$. Under these assumptions, it is possible to prove that the proposed algorithm converges to a stationary point of problem $[P.5]$ (see \cite{Scutari_ICASSP14},\cite{Scutari_nonconvex} for details on the convergence proof).

\begin{algorithm}[t]

\begin{flushleft}

\normalsize

\smallskip
\textbf{Data:} $\by^0\triangleq(\mathbf{I}^{0},\bt^0,\bp^0)\in \mathcal{X}$; $\{\beta^{\nu}\}_\nu \in (0,1]$; $\{\delta^\nu\}_\nu$; $\tau_p>0$; $\tau_I>0$, $\tau_t>0$. Set $\nu=0$.

\vspace{.15cm}
(\texttt{S.1}): If $\by^{\nu}$ satisfies a termination criterion, \texttt{STOP};

\vspace{.15cm}
(\texttt{S.2}): Compute the solution $\hat{\by}(\by^{\nu})$ of problem $[{P.6}]$;

\vspace{.15cm}
(\texttt{S.3}):  Set $\by^{\nu+1}=\by^{\nu}+\beta^{\nu}\left(\hat{\by}(\by^{\nu})-\by^{\nu}\right)$;

\vspace{.15cm}
(\texttt{S.4}):   $\nu \leftarrow \nu+1$  and go  to (\texttt{S.1}).

\vspace{.15cm}
\caption{\textbf{:}  Inner SCA  Algorithm for $[{P.5}]$  \label{alg:Alg_centr}}

\end{flushleft}

\end{algorithm}

In the sequel, we provide some numerical results aimed at illustrating the performance of the proposed relaxed method when compared to the global optimal solution provided by Algorithm 1.

\subsection{Numerical Results}

The method proposed in Algorithm 2 performs a relaxation of the integer variables, which are now real variables that can assume any value between zero and one. It is then fundamental to check if the proposed method is still able to give meaningful results that well approximate the solutions of the original problem $[P.1]$. In particular, we have noticed that, despite the relaxation, in many cases the results achieved through algorithm 2 provide values for the integer variables that are exactly equal to zero or one. In some other cases, the final results are not exactly integer, but still very close to zero and one. In these latter cases, we associate the final result obtained through Algorithm 2 with the closest combination of the call graph, and find the optimal transmit power by solving the correspondent convex radio allocation problem in $[P.2]$ (i.e. $[P.7]$). Furthermore, it is also important to assess the convergence speed of Algorithm 2, because a true complexity reduction with respect to $[P.1]$ is obtained only if the method is able to converge in a few iterations. As an example, in Fig. \ref{convergence} we report the behavior of Algorithm 2 in terms of energy spent for transmission/processing versus the iteration index, considering three different initializations of the algorithm. The radio parameters and the call graph are the same used to achieve the results in Fig. 3, and the equivalent channel $a$ is set equal to 500. The step-size is chosen in order to satisfy the rule $\beta^{\nu}=\beta^{\nu-1}(1-\mu\beta^{\nu-1})$, with $\beta^{0}=0.2$, and $\mu=10^{-4}$. The sequence $\delta^\nu$ is chosen such that  $\delta^\nu=\delta^0/(\nu+1)$, with $\delta^0=10^{-4}$. As we can notice from Fig. \ref{convergence}, for all the initializations, the algorithm converges to the same solution in a few iterations. In particular, the solution found by the proposed algorithm coincides with the optimal solution obtained by Algorithm 1.

\begin{figure}[t]
\centering
\includegraphics[width=7.5cm]{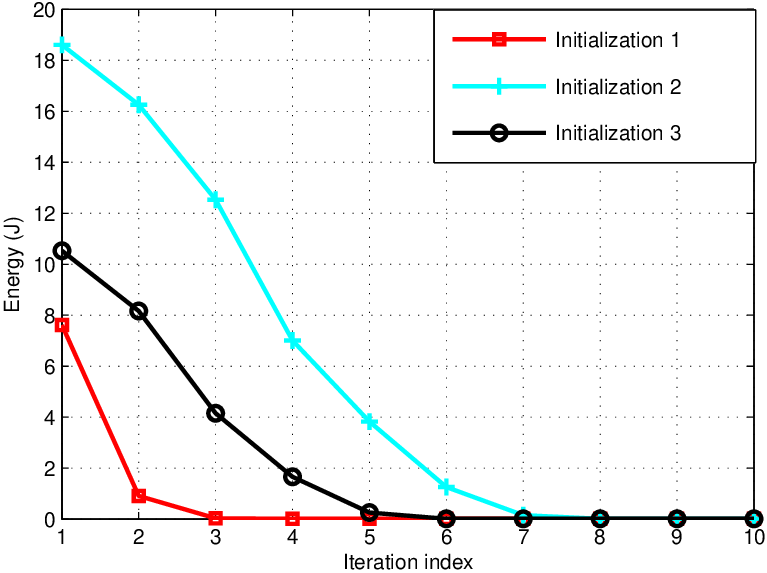}
  \caption{Behavior of Energy versus iteration index, for different initializations of Algorithm 2.}\label{convergence}
\end{figure}

To assess the performance of Algorithm 2 (SCA) with respect to the global optimum obtained by Algorithm 1, in Fig. \ref{SCA}, we report the average energy spent for processing versus the distance between MUE and SCceNB for both algorithms. We consider different communication strategies, i.e. SISO, and 1x2 SIMO. The results are averaged over 200 independent realizations. The parameters are the same used in the previous simulation. From Fig. \ref{SCA}, as expected, the energy spent for processing increases at larger distances, and the possibility of having a MIMO channel improves the overall performance. Interestingly, we can see from Fig. \ref{SCA} that the performance of the relaxed method based on SCA is very close to the optimal solution achieved by Algorithm 1. This result illustrates how Algorithm 2 might be a good, low complexity alternative to Algorithm 1, especially when the size of the call graph makes impossible an exhaustive search over all the combinations.

\begin{figure}[t]
\centering
\includegraphics[width=7.5cm]{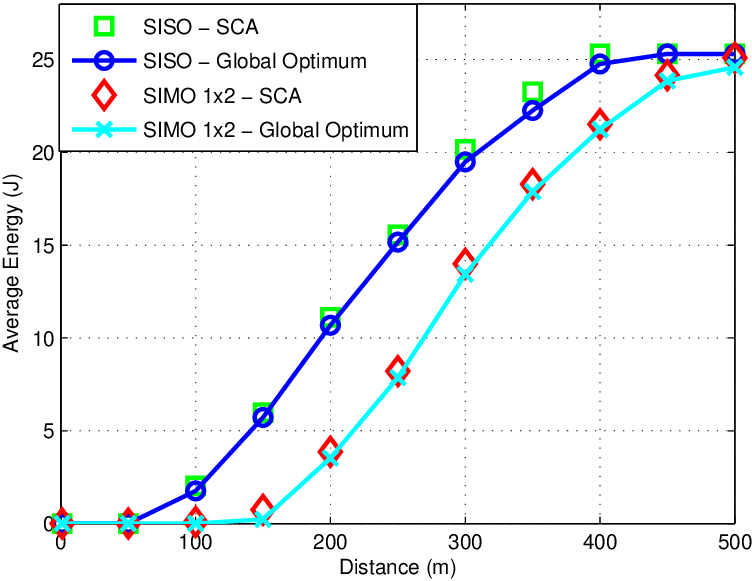}
  \caption{Average energy spent for processing versus distance (from MUE to SCceNB), for different algorithms and communication strategies.}\label{SCA}
\end{figure}

\section{Conclusions}

In this paper, we have proposed a method to optimize the allocation of communication resources and the call graph's partition of a computer program jointly, in a mobile edge computing context, with the aim of minimizing the energy consumption at the mobile user, while satisfying a delay constraint imposed by the application. We have proved that, for any given call graph's partition, the optimization problem associated to the selection of the  radio parameters only is convex and then it can be solved with numerically efficient methods. However, the overall problem turns out to have a combinatorial complexity that increases with the granularity of the call graph structure.  To cope with the combinatorial nature of the problem, we have first illustrated how several graph configurations can be discarded depending on the current channel value, thus reducing the complexity burden of the overall approach. Then, we have introduced a lower complexity method that solves a relaxed version of the original optimization problem, which is shown to perform very similarly to the global optimal solution. Simulation results have been presented to illustrate for what kind of application and  channel conditions, computation offloading can provide a significant performance gain with respect to local computation. In this work, we assumed the parameters of the call graph to be known. An interesting development of this work is the incorporation of learning mechanisms to tackle the case where these parameters are to be estimated or predicted. A further development concerns the multiuser case, where a set of mobile users concur for the use of the same radio/computation resource pool, while incorporating graph partitioning strategies.

\appendices

\section{Proof of Theorem 1}

In the following, for simplicity of notation and without any loss of generality, we assume that the links $(u,v)\in \mathcal{E}^c_e$ are identified by the index $i=1,\ldots,{\rm card}(\mathcal{E}_e^c)$.

To prove point (a), which ensures that the feasibility set is nonempty, we note that by inverting the relation
\begin{equation}
\sum_{i=1}^{{\rm card}(\mathcal{E}_e^c)} \frac{N'_i}{\displaystyle\log\left(1+a P_T\right)}\leq L_c,
\end{equation}
which is the first constraint in $[P.2]$ where instead of the $p_i$'s we have considered the maximum power value $P_T$, we get the inequality (\ref{feasibility}), which is a sufficient condition for the existence of a nonempty feasible set.

We proceed now in proving point (b). As said before, problem $[P.2]$ is nonconvex. However, let us consider the change of variables
\begin{equation}\label{var_change}
t_i=\log(1+ap_i) \; \rightarrow \; p_i=\frac{e^{t_i}-1}{a},
\end{equation}
$\forall\; i=1,\ldots,{\rm card}(\mathcal{E}_e^c)$. Relation (\ref{var_change}) is a one to one mapping such that it is always possible to find uniquely the variable $p_i$ from $t_i$ and viceversa. Thus, using (\ref{var_change}), the optimization problem in $[P.2]$ becomes equivalent to the following problem (see \cite[p. 130]{Boyd}):
\begin{align} \label{Opt_problem3}
[{P.7}]\hspace{.5cm}\min_{\displaystyle \bt_c} & \quad \sum_{i=1}^{\mathcal{E}^c_e}\frac{N'_{i}}{a}\frac{e^{t_i}-1}{t_i} \nonumber\\
& \hbox{s.t.} \quad \sum_{i=1}^{{\rm card}(\mathcal{E}_e^c)} \frac{N'_i}{t_i}\leq L_c \nonumber\\
& \quad \quad \; 0< t_i \leq T_{\max}, \quad i=1,\ldots,{\rm card}(\mathcal{E}_e^c),\nonumber
\end{align}
where $\bt_c=[t_1,\ldots,t_{{\rm card}(\mathcal{E}^c_e)}]^T$, and $T_{\max}=\log(1+aP_T)$. It is now straightforward to check the convexity of the optimization set, which is given by the intersection of convex sets. We now have to prove the strict convexity of the objective function. Let us define
\begin{equation}\label{obj_function}
f(\bt_c)=\sum_{i=1}^{{\rm card}(\mathcal{E}^c_e)}N''_{i}\frac{e^{t_i}-1}{t_i}
\end{equation}
where $N''_{i}=N'_{i}/a>0$. Since the function in (\ref{obj_function}) is separable in the optimization variables, its Hessian is a diagonal matrix.
Thus, to show the strict convexity of function (\ref{obj_function}), it is sufficient to prove that each diagonal element of the Hessian matrix is strictly positive inside the optimization set. Let us consider the $i$-th component. The second order partial derivative is given by
\begin{align}\label{second_derivative}
\frac{\partial^2 f(\bt_c)}{\partial^2 t_i} &= N''_{i}\frac{t_i^2e^{t_i}-2e^{t_i}(t_i-1)-2}{t^3_i}=N''_{i}\frac{g(t_i)-2}{t^3_i}
\end{align}
where $g(t_i)=t_i^2e^{t_i}-2e^{t_i}(t_i-1)=e^{t_i}(t_i-1)^2+e^{t_i}$. Since $t_i$ and $N''_{i}$ are positive, we have to analyze only the positiveness of the function $g(t_i)-2$. It is easy to verify that
$\lim_{t_i\rightarrow 0^+} g(t_i)=2$,
where the notation $\displaystyle \lim_{t_i\rightarrow t_0^+} h(t)$ means the right limit of function $h(t)$ around the point $t_0$. Then, if we prove that $g(t_i)$ is monotonic increasing for $t_i>0$, we have showed that the second derivative in (\ref{second_derivative}) is strictly positive. In particular, we have
\begin{equation}
\frac{dg(t_i)}{dt_i}=e^{t_i}\left[(t_i-1)^2+1\right]+2(t_i-1)e^{t_i}=t_i^2e^{t_i}>0,
\end{equation}
$\forall t_i>0$. Since the previous arguments hold true for all $i\in[1,\ldots,{\rm card}(\mathcal{E}_e^c)]$, we can conclude that the Hessian is positive definite and the objective function (\ref{obj_function}) is strictly convex, thus concluding the proof of point (b).

\section{Proof of Theorem 2}

To prove point (a), let us consider the two constraints of problem $[P.4]$.  For simplicity of notation let us  assume that the links $(u,v)\in \mathcal{E}_e^c$ are identified by the index $i=1,\ldots,{\rm card}(\mathcal{E}_e^c)$. To get a sufficient condition, which guarantees that the delay constraint in $[P.4]$ is satisfied, we must find a lower bound of the first constraint with respect to all the possible power allocations $\{p_i^k\}$, subject to the presence of the power budget constraints that compose the second constraint. This can be achieved by maximizing the denominator of the first constraint, for all $i$, subject to the presence of the second constraint in $[P.4]$. Thus, since the power budget constraints in $[P.4]$ are the same for all $i$, the feasibility condition is given by (\ref{feasibility2}), where $\bp^*=\{p^*_k\}_{k=1}^K$ is the solution of the problem:
\begin{align}\label{WF}
\max_{\displaystyle \bp}  \; \displaystyle\sum_{k=1}^K\log\left(1+a_k p^k\right) \quad \hbox{s.t.} \;\; \sum_{k=1}^K p^k \leq P_T.
\end{align}
The solution $\bp^*=\{p^*_k\}_{k=1}^K$ of (\ref{WF}) is the well known Water Filling \cite{Proakis}.
Using $\bp^*$, it is clear that the inequality (\ref{feasibility2}) holds, thus giving a sufficient condition on the channel values ensuring the existence of a non-empty feasible set.

To prove point (b), we exploit arguments that are similar to what we have used in Appendix A. In particular, let us consider the following change of variables:
\begin{align}\label{var_change2}
t^k_i=\log(1+a_kp^k_i) \; \rightarrow \; p^k_i=\left(e^{t^k_i}-1\right)/a_k,
\end{align}
$\forall\; i=1,\ldots,{\rm card}(\mathcal{E}_e^c), \;\;\forall\; k=1,\ldots,K$. Relation (\ref{var_change2}) is a one to one mapping such that it is always possible to find uniquely the variable $p^k_i$ from $t^k_i$ and viceversa. Thus, using (\ref{var_change2}), the optimization problem in $[P.4]$ becomes equivalent to the following problem:

\begin{align} \label{Opt_problem6}
\hspace{-.3cm}[{P.8}]\hspace{.3cm} \min_{\displaystyle \bt_c} & \quad \sum_{i=1}^{{\rm card}(\mathcal{E}_e^c)}\frac{\displaystyle N'_{i}\sum_{k=1}^K \frac{e^{t^k_i}-1}{a_k}}{\displaystyle\sum_{k=1}^K t^k_i} \nonumber\\
&\hspace{-.5cm} \hbox{s.t.} \quad \sum_{i=1}^{{\rm card}(\mathcal{E}_e^c)} \frac{N'_i}{\displaystyle\sum_{k=1}^K t^k_i}\leq L'_c \nonumber\\
&\hspace{-.5cm} \quad \quad \; \sum_{k=1}^K \frac{e^{t^k_i}-1}{a_k}\leq P_T, \quad i=1,\ldots,{\rm card}(\mathcal{E}_e^c) \nonumber
\end{align}
where $\bt_c=[\bt_1,\ldots,\bt_{{\rm card}(\mathcal{E}_e^c)}]^T$ and $\bt_i=[t^1_i,\ldots, t^K_i]$. We first prove the convexity of the set. It is straightforward to see that all the constraints grouped in the second constraint in $[P.8]$ are convex. Furthermore, let us define
\begin{equation}
g(\bt_c)=\sum_{i=1}^{{\rm card}(\mathcal{E}_e^c)}\frac{N'_i}{\displaystyle\sum_{k=1}^K t^k_i}= \sum_{i=1}^{{\rm card}(\mathcal{E}_e^c)}g_i(\bt_i).
\end{equation}
Since each function $g_i(\bt_i)$ is given by the composition of a convex function and an affine mapping, $g_i(\bt_i)$ is a convex function (see \cite{Boyd}) for all $i=1,\ldots,{\rm card}(\mathcal{E}_e^c)$. The overall sum function $g(\bt_c)$ is then convex, thus making the first constraint in $[P.8]$ a convex set.
The optimization set is then convex since it is given by the intersection of convex sets. Consider now the objective function of $[P.8]$, which reads as:
\begin{equation}\label{f}
f(\bt_c)=\sum_{i=1}^{{\rm card}(\mathcal{E}_e^c)}\frac{\displaystyle N'_{i}\sum_{k=1}^K \frac{e^{t^k_i}-1}{a_k}}{\displaystyle\sum_{k=1}^K t^k_i}=\sum_{i=1}^{{\rm card}(\mathcal{E}_e^c)} w_i(\bt_i).
\end{equation}
Before proceeding with the proof, we introduce some definitions on generalized convexity which are instrumental to prove that every stationary point is a global minimum. A differentiable function $f(\bt_c)$ is pseudoconvex at a point $\bt^{0}_c$ \cite{Mangasarian}, if
\beq
(\bt_c-\bt^{0}_c)^T \nabla f(\bt^{0}_c)\geq 0 \quad \Rightarrow \quad f(\bt_c)\geq f(\bt^{0}_c) \label{pseud_def}
\eeq
for any $\bt_c \in \mbox{dom}\, f$. In \cite{Mangasarian}[pag. 142], it is proved that, if $f$ is pseudoconvex at $\bt^{0}_c$, then having $\nabla f(\bt^{0}_c)=0$ implies that $\bt^{0}_c$  is a globally optimal point. Hence, if $f$ is pseudoconvex at every stationary point, then every stationary point is a global minimum. To exploit this property, we need to prove that $f(\bt_c)$  is  pseudoconvex at every stationary point. More specifically,  consider $f(\bt_c)=\ds \sum_{i=1}^{{\rm card}(\mathcal{E}_e^c)}w_i(\bt_i)= \ds \sum_{i=1}^{{\rm card}(\mathcal{E}_e^c)} \ds \frac{h(\bt_i)}{v(\bt_i)}$
where $h(\bt_i)=\displaystyle N'_{i}\sum_{k=1}^K \frac{e^{t^k_i}-1}{a_k}$ and $v(\bt_i)=\displaystyle\sum_{k=1}^K t^k_i$.
Let us now show that each function $w_i(\bt_i)$ is pseudoconvex  $\forall \, \bt_i \in \mbox{dom}\, w_i$.
Assume that $(\bt_i-\bt^{0}_i)^T \nabla w_i(\bt^{0}_i)\geq 0$, i.e.
\beq
(\bt_i-\bt^{0}_i)^T \left(\ds \frac{\nabla h(\bt^0_i)}{v(\bt^{0}_i)}- \ds \frac{h(\bt^0_i)}{v^2(\bt^0_i)}\ds \nabla v(\bt^0_i)\right) \geq 0.
\label{grad_g1}
\eeq
Since $h(\bt_i)$ is a differentiable convex function and $v(\bt_i)$ is linear, the
following conditions hold:
\begin{align}
 h(\bt_i) &\geq  h(\bt^{0}_i)+ \nabla^T h(\bt^{0}_i) (\bt_i-\bt^{0}_i) \\
 v(\bt_i) &= v(\bt^{0}_i)+ \nabla^T v(\bt^{0}_i)(\bt_i-\bt^{0}_i),
\end{align}
for all $\bt^{0}_i, \bt_i  \in \mbox{dom}\, w_i$. Hence, by using these  inequalities in (\ref{grad_g1}) and since $v(\bt_i)>0$, it results
\beq
\begin{array}{lll}
\left(\ds \frac{h(\bt_i)}{v(\bt^{0}_i)}-
 \ds \frac{h(\bt^{0}_i)}{v^2(\bt^{0}_i)}v(\bt_i) \right) \geq 0
\label{grad_g}
\end{array}
\eeq
i.e.
\beq
w_i(\bt_i)=\frac{h(\bt_i)}{v(\bt_i)}\geq \frac{h(\bt^{0}_i)}{v(\bt^{0}_i)}=w_i(\bt^{0}_i)
\label{wi}
\eeq
for any $\bt_i, \bt^{0}_i \in \mbox{dom} \, w_i$. Then, from (\ref{wi}) and (\ref{pseud_def}), we can state that $w_i(\bt_i)$ is a pseudoconvex function for every $i=1,\ldots, {{\rm card}(\mathcal{E}_e^c)}$. Let us now prove that  $f(\bt_c)= \sum_{i=1}^{{\rm card}(\mathcal{E}_e^c)} w_i(\bt_i)$ is a pseudoconvex function at every stationary point $\bt^{0}_c$. Hence assume that  $(\bt_c-\bt^{0}_c)^T \nabla f(\bt^{0}_c)\geq 0$,
where $\bt^{0}_c=[\bt^{0}_1,\ldots, \bt^{0}_{{\rm card}(\mathcal{E}_e^c)}]$ is a stationary point of $f(\bt_c)$.
Let us denote with $\mathcal{S}=\{\bt^{0}_c \in \mbox{dom} f : \; \nabla f(\bt^{0}_c)=0 \}$ the set of the stationary points of  $f(\bt_c)$.
Observe that since the functions $w_i(\bt_i)$ are uncoupled, the set $\mathcal{S}$
 is equal to the cartesian product of the sets  $\mathcal{S}_i$ of the stationary points of each function $w_i(\bt_i)$, i.e. $\mathcal{S}=\mathcal{S}_1\times \mathcal{S}_2 \times \ldots \times \mathcal{S}_{{\rm card}(\mathcal{E}^c_e)}$.
 Hence from the pseudoconvexity of each function $w_i(\bt_i)$ we can state that
 \beq
 (\bt_i-\bt^{0}_i)^T \nabla w_i(\bt^{0}_i)=0 \; \Rightarrow \; w_i(\bt_i)\geq w_i(\bt^{0}_i) \; \forall \bt_i \in \mbox{dom} \;w_i
 \eeq
so that $ \ds \sum_{i=1}^{{\rm card}(\mathcal{E}_e^c) }w_i(\bt_i)\geq \ds \sum_{i=1}^{{\rm card}(\mathcal{E}_e^c) } w_i(\bt^{0}_i)$.
Hence, at each stationary point $\bt^{0}_c$, we get
\beq
(\bt_c-\bt^{0}_c)^T \nabla f(\bt^{0}_c)=0\; \Rightarrow f(\bt_c)\geq f(\bt^{0}_c)
\eeq
for any $\bt_c \in \mbox{dom}\,f$. Then, from (\ref{pseud_def}), $f$ is pseudoconvex at every stationary point, thus ensuring that every stationary point of $f$ is a global minimum. This completes the proof of point (b).



\balance

\bibliographystyle{MyIEEE}
\bibliography{bmc_article}

\begin{thebibliography}{10}
\providecommand{\url}[1]{#1}
\csname url@samestyle\endcsname
\providecommand{\newblock}{\relax}
\providecommand{\bibinfo}[2]{#2}
\providecommand{\BIBentrySTDinterwordspacing}{\spaceskip=0pt\relax}
\providecommand{\BIBentryALTinterwordstretchfactor}{4}
\providecommand{\BIBentryALTinterwordspacing}{\spaceskip=\fontdimen2\font plus
\BIBentryALTinterwordstretchfactor\fontdimen3\font minus
  \fontdimen4\font\relax}
\providecommand{\BIBforeignlanguage}[2]{{%
\expandafter\ifx\csname l@#1\endcsname\relax
\typeout{** WARNING: IEEEtran.bst: No hyphenation pattern has been}%
\typeout{** loaded for the language `#1'. Using the pattern for}%
\typeout{** the default language instead.}%
\else
\language=\csname l@#1\endcsname
\fi
#2}}
\providecommand{\BIBdecl}{\relax}
\BIBdecl

\bibitem{Lewin1}
{\it iPhone Users 30 Times More Likely To Watch YouTube Videos.} {A}vailable
  at: http://www.podcastingnews
  .com/2008/03/19/iphone-users-30-times-watch-youtube-videos/.

\bibitem{Powers}
R.~Powers, ``Batteries for low power electronics,'' \emph{Proceedings of the
  IEEE}, vol.~83, 1995.

\bibitem{Palacin}
M.~R. Palacin, ``Recent advances in rechargeable battery materials: a
  chemist’s perspective,'' \emph{Chemical Society Reviews}, vol.~38, no.~9,
  pp. 2565--2575, 2009.

\bibitem{Sharifi}
M.~Sharifi, S.~Kafaie, and O.~Kashefi, ``A survey and taxonomy of cyber
  foraging of mobile devices,'' \emph{IEEE Communications Surveys \&
  Tutorials}, vol.~14, no.~4, pp. 1232--1243, 2012.

\bibitem{Kumar-Liu-Lu-Bhargava}
K.~Kumar, J.~Liu, Y.-H. Lu, and B.~Bhargava, ``A survey of computation
  offloading for mobile systems,'' \emph{Mobile Networks and Applications},
  vol.~18, no.~1, pp. 129--140, 2013.

\bibitem{Hayes}
B.~Hayes, ``Cloud computing,'' \emph{Communications of the ACM}, vol.~51,
  no.~7, 2008.

\bibitem{Dinh-Lee-Niyato-Wang}
H.~T. Dinh, C.~Lee, D.~Niyato, and P.~Wang, ``A survey of mobile cloud
  computing: architecture, applications, and approaches,'' \emph{Wireless
  communications and mobile computing}, vol.~13, no.~18, pp. 1587--1611, 2013.

\bibitem{Fernando-Loke-Rahayu}
N.~Fernando, S.~W. Loke, and W.~Rahayu, ``Mobile cloud computing: A survey,''
  \emph{Future Generation Computer Systems}, vol.~29, no.~1, pp. 84--106, 2013.

\bibitem{SB-SS-PDL2}
S.~Barbarossa, S.~Sardellitti, and P.~Di~Lorenzo, ``Communicating while
  computing: Distributed mobile cloud computing over {5G} heterogeneous
  networks,'' \emph{IEEE Signal Processing Magazine}, vol.~31, no.~6, pp.
  45--55, 2014.

\bibitem{Li-Wang-Xu}
Z.~Li, C.~Wang, and R.~Xu, ``Computation offloading to save energy on handheld
  devices: a partition scheme,'' in \emph{Proceedings of the ACM international
  conference on Compilers, architecture, and synthesis for embedded systems},
  2001, pp. 238--246.

\bibitem{Wang-Li}
C.~Wang and Z.~Li, ``Parametric analysis for adaptive computation offloading,''
  in \emph{ACM SIGPLAN Notices}, vol.~39, no.~6, 2004, pp. 119--130.

\bibitem{Yang-Ou-Chen}
K.~Yang, S.~Ou, and H.-H. Chen, ``On effective offloading services for
  resource-constrained mobile devices running heavier mobile internet
  applications,'' \emph{IEEE Communications Magazine}, vol.~46, no.~1, pp.
  56--63, 2008.

\bibitem{Chun-Maniatis}
B.-G. Chun and P.~Maniatis, ``Augmented smartphone applications through clone
  cloud execution,'' in \emph{HotOS}, vol.~9, 2009, pp. 8--11.

\bibitem{Chun-Ihm-Maniatis-Naik}
B.-G. Chun, S.~Ihm, P.~Maniatis, and M.~Naik, ``Clonecloud: boosting mobile
  device applications through cloud clone execution,'' \emph{arXiv preprint
  arXiv:1009.3088}, 2010.

\bibitem{Tang}
L.~Tang and Q.~Li, ``Energy and time optimization for wireless computation
  offloading,'' in \emph{IEEE International Conference on Wireless
  Communications}, 2015, pp. 1--5.

\bibitem{Kwon}
Y.~Kwon, H.~Yi, D.~Kwon, S.~Yang, Y.~Cho, and Y.~Paek, ``Precise execution
  offloading for applications with dynamic behavior in mobile cloud
  computing,'' \emph{Pervasive and Mobile Computing}, 2015.

\bibitem{Gao-He-Liu-Li-Jarvis}
B.~Gao, L.~He, L.~Liu, K.~Li, and S.~A. Jarvis, ``From mobiles to clouds:
  Developing energy-aware offloading strategies for workflows,'' in
  \emph{Proceedings of the ACM/IEEE International Conference on Grid
  Computing}, 2012, pp. 139--146.

\bibitem{Wen-Zhang-Luo}
Y.~Wen, W.~Zhang, and H.~Luo, ``Energy-optimal mobile application execution:
  Taming resource-poor mobile devices with cloud clones,'' in \emph{Proceedings
  of IEEE INFOCOM}, 2012, pp. 2716--2720.

\bibitem{Barb-Sard-Dilo}
S.~Barbarossa, S.~Sardellitti, and P.~Di~Lorenzo, ``Computation offloading for
  mobile cloud computing based on wide cross-layer optimization,'' in
  \emph{Proc. of Future Network and Mobile Summit}, 2013, pp. 1--10.

\bibitem{Barb-Dilo-Sard}
S.~Barbarossa, P.~Di~Lorenzo, and S.~Sardellitti, ``Computation offloading
  strategies based on energy minimization under computational rate
  constraints,'' in \emph{European Conference on Networks and Communications},
  2014, pp. 1--5.

\bibitem{Wolsky-Gurun-Krintz-Nurmi}
R.~Wolski, S.~Gurun, C.~Krintz, and D.~Nurmi, ``Using bandwidth data to make
  computation offloading decisions,'' in \emph{IEEE International Symposium on
  Parallel and Distributed Processing}, 2008, pp. 1--8.

\bibitem{Kao}
Y.-H. Kao and B.~Krishnamachari, ``Optimizing mobile computational offloading
  with delay constraints,'' in \emph{Proc. of the IEEE Global Communications
  Conference}, 2014, pp. 2289--2294.

\bibitem{Maui}
E.~Cuervo, A.~Balasubramanian, D.-k. Cho, A.~Wolman, S.~Saroiu, R.~Chandra, and
  P.~Bahl, ``Maui: making smartphones last longer with code offload,'' in
  \emph{Proceedings of the ACM international conference on Mobile systems,
  applications, and services}, 2010, pp. 49--62.

\bibitem{Chen-Kang-Kandemir-Vijaykrishnan-Irwin-Chandramouli}
G.~Chen, B.-T. Kang, M.~Kandemir, N.~Vijaykrishnan, M.~J. Irwin, and
  R.~Chandramouli, ``Studying energy trade offs in offloading
  computation/compilation in java-enabled mobile devices,'' \emph{IEEE
  Transactions on Parallel and Distributed Systems}, vol.~15, no.~9, pp.
  795--809, 2004.

\bibitem{Hong-Kumar-Lu}
Y.-J. Hong, K.~Kumar, and Y.-H. Lu, ``Energy efficient content-based image
  retrieval for mobile systems,'' in \emph{IEEE International Symposium on
  Circuits and Systems}, 2009, pp. 1673--1676.

\bibitem{Nimmagadda-Kumar-Lee}
Y.~Nimmagadda, K.~Kumar, Y.-H. Lu, and C.~G. Lee, ``Real-time moving object
  recognition and tracking using computation offloading,'' in \emph{IEEE
  International Conference on Intelligent Robots and Systems}, 2010, pp.
  2449--2455.

\bibitem{Ou-Yang-Liotta-Lu}
S.~Ou, K.~Yang, A.~Liotta, and L.~Hu, ``Performance analysis of offloading
  systems in mobile wireless environments,'' in \emph{IEEE International
  Conference on Communications}, 2007, pp. 1821--1826.

\bibitem{Xian-Lu-Li}
C.~Xian, Y.-H. Lu, and Z.~Li, ``Adaptive computation offloading for energy
  conservation on battery-powered systems,'' in \emph{IEEE International
  Conference on Parallel and Distributed Systems}, vol.~2, 2007, pp. 1--8.

\bibitem{Liu}
J.~Liu, E.~Ahmed, M.~Shiraz, A.~Gani, R.~Buyya, and A.~Qureshi, ``Application
  partitioning algorithms in mobile cloud computing: Taxonomy, review and
  future directions,'' \emph{Journal of Network and Computer Applications},
  vol.~48, pp. 99--117, 2015.

\bibitem{Yang13}
L.~Yang, J.~Cao, Y.~Yuan, T.~Li, A.~Han, and A.~Chan, ``A framework for
  partitioning and execution of data stream applications in mobile cloud
  computing,'' \emph{ACM Sigmetrics Performance Evaluation Review}, vol.~40,
  no.~4, pp. 23--32, 2013.

\bibitem{Huerta-Canepa-Lee}
G.~Huerta-Canepa and D.~Lee, ``An adaptable application offloading scheme based
  on application behavior,'' in \emph{IEEE International Conference on Advanced
  Information Networking and Applications-Workshops}, 2008, pp. 387--392.

\bibitem{Kumar-Lu}
K.~Kumar and Y.-H. Lu, ``Cloud computing for mobile users: Can offloading
  computation save energy?'' \emph{IEEE Computer}, no.~4, pp. 51--56, 2010.

\bibitem{Balan}
R.~K. Balan, ``Powerful change part 2: reducing the power demands of mobile
  devices,'' \emph{IEEE Pervasive Computing}, vol.~3, no.~2, pp. 71--73, 2004.

\bibitem{Chen-Kang-Kandemir}
G.~Chen, B.-T. Kang, M.~Kandemir, N.~Vijaykrishnan, M.~J. Irwin, and
  R.~Chandramouli, ``Studying energy trade offs in offloading
  computation/compilation in java-enabled mobile devices,'' \emph{IEEE
  Transactions on Parallel and Distributed Systems}, vol.~15, no.~9, pp.
  795--809, 2004.

\bibitem{Kovachev-Yu-Klamma}
D.~Kovachev, T.~Yu, and R.~Klamma, ``Adaptive computation offloading from
  mobile devices into the cloud,'' in \emph{IEEE International Symposium on
  Parallel and Distributed Processing with Applications}, 2012, pp. 784--791.

\bibitem{Huang-Wang-Niyato}
D.~Huang, P.~Wang, and D.~Niyato, ``A dynamic offloading algorithm for mobile
  computing,'' \emph{IEEE Transactions on Wireless Communications}, vol.~11,
  no.~6, pp. 1991--1995, 2012.

\bibitem{Kosta_2012}
S.~Kosta, A.~Aucinas, P.~Hui, R.~Mortier, and X.~Zhang, ``Thinkair: Dynamic
  resource allocation and parallel execution in the cloud for mobile code
  offloading,'' in \emph{Proceedings IEEE INFOCOM}, 2012, pp. 945--953.

\bibitem{Xia2014}
F.~Xia, F.~Ding, J.~Li, X.~Kong, L.~T. Yang, and J.~Ma, ``Phone2cloud:
  Exploiting computation offloading for energy saving on smartphones in mobile
  cloud computing,'' \emph{Information Systems Frontiers}, vol.~16, no.~1, pp.
  95--111, 2014.

\bibitem{SB-SS-PDL}
S.~Barbarossa, S.~Sardellitti, and P.~Di~Lorenzo, ``Joint allocation of
  computation and communication resources in multiuser mobile cloud
  computing,'' in \emph{Proc. of IEEE International Workshop on Signal
  Processing Advances in Wireless Communication}, 2013, pp. 26--30.

\bibitem{Labidi}
W.~Labidi, M.~Sarkiss, and M.~Kamoun, ``Joint multi-user resource scheduling
  and computation offloading in small cell networks,'' in \emph{IEEE
  International Conference on Wireless and Mobile Computing, Networking and
  Communications}, 2015, pp. 794--801.

\bibitem{Sard-Scut-Barb}
S.~Sardellitti, G.~Scutari, and S.~Barbarossa, ``Joint optimization of radio
  and computational resources for multicell mobile-edge computing,'' \emph{IEEE
  Transactions on Signal and Information Processing over Networks}, vol.~1,
  no.~2, pp. 89--103, 2015.

\bibitem{Chen}
X.~Chen, ``Decentralized computation offloading game for mobile cloud
  computing,'' \emph{IEEE Transactions on Parallel and Distributed Systems},
  vol.~26, no.~4, pp. 974--983, 2015.

\bibitem{Cardellini}
V.~Cardellini, V.~D.~N. Person{\'e}, V.~Di~Valerio, F.~Facchinei, V.~Grassi,
  F.~L. Presti, and V.~Piccialli, ``A game-theoretic approach to computation
  offloading in mobile cloud computing,'' \emph{Mathematical Programming}, pp.
  1--29, 2013.

\bibitem{Satyanarayanan-Bahl-Caceres-Davies}
M.~Satyanarayanan, P.~Bahl, R.~Caceres, and N.~Davies, ``The case for vm-based
  cloudlets in mobile computing,'' \emph{IEEE Pervasive Computing}, vol.~8,
  no.~4, pp. 14--23, 2009.

\bibitem{Barbera-Kosta-Mei-Stefa}
M.~Barbera, S.~Kosta, A.~Mei, and J.~Stefa, ``To offload or not to offload? the
  bandwidth and energy costs of mobile cloud computing,'' in \emph{Proceedings
  IEEE INFOCOM}, 2013, pp. 1285--1293.

\bibitem{TROPIC}
{TROPIC}: Distributed computing, storage and radio resource allocation over
  cooperative femtocells. http://www.ict-tropic.eu.

\bibitem{ETSI}
{ETSI} first meeting of new standardization group on mobile-edge computing.
  http://www.etsi.org/newsevents/news/838-2014-10-news-etsi-announces-first-meeting-of-newstandardization-group-on-mobile-edge-computing.

\bibitem{Ryder}
B.~G. Ryder, ``Constructing the call graph of a program,'' \emph{IEEE
  Transactions on Software Engineering}, no.~3, pp. 216--226, 1979.

\bibitem{Grove-Chambers}
D.~Grove and C.~Chambers, ``A framework for call graph construction
  algorithms,'' \emph{ACM Transactions on Programming Languages and Systems},
  vol.~23, no.~6, pp. 685--746, 2001.

\bibitem{Scutari_ICASSP14}
G.~Scutari, F.~Facchinei, L.~Lampariello, and P.~Song, ``Parallel and
  distributed methods for nonconvex optimization,'' in \emph{IEEE International
  Conference on Acoustics, Speech and Signal Processing}, 2014, pp. 840--844.

\bibitem{Scutari_nonconvex}
------, ``Parallel and distributed methods for nonconvex optimization -- part
  {I}: Theory,'' \emph{http://arxiv.org/abs/1410.4754}, 2014.

\bibitem{Nemhauser-Wolsey}
L.~A. Wolsey and G.~L. Nemhauser, \emph{Integer and combinatorial
  optimization}.\hskip 1em plus 0.5em minus 0.4em\relax New Jersey: John Wiley
  and Sons, 2014.

\bibitem{Boyd}
S.~Boyd and L.~Vandenberghe, \emph{Convex optimization}.\hskip 1em plus 0.5em
  minus 0.4em\relax New York: Cambridge university press, 2004.

\bibitem{Barbarossa}
S.~Barbarossa, \emph{Multiantenna wireless communications systems}.\hskip 1em
  plus 0.5em minus 0.4em\relax Massachusetts: Artech House Publishers, 2005.

\bibitem{3GPP}
3gpp tr 36.814, ``technical specification group radio access network; further
  advancements for e-utra, physical layer aspects,'' release 9, v.9.0.0., march
  2010.

\bibitem{Proakis}
J.~Proakis, \emph{Digital Communications}.\hskip 1em plus 0.5em minus
  0.4em\relax New York: McGraw-Hill, 2001.

\bibitem{Mangasarian}
O.~L. Mangasarian, \emph{Nonlinear programming}.\hskip 1em plus 0.5em minus
  0.4em\relax Philadelphia: Society for Industrial and Applied Mathematics,
  1994.

\end{thebibliography}

\end{document}